\title{Weighted $k$-Server Bounds via Combinatorial Dichotomies
\thanks{This work was supported by NWO grant 639.022.211 and ERC consolidator
grant 617951.}}

\author{
Nikhil Bansal\thanks{TU Eindhoven, Netherlands.
\texttt{\{n.bansal,m.elias,g.koumoutsos\}@tue.nl}},
Marek Eli\'a\v{s}\footnotemark[2],
Grigorios Koumoutsos\footnotemark[2]}

\documentclass[a4paper,11pt]{article}

\usepackage{fullpage}
\usepackage{amsmath,amssymb,amsthm,enumerate}
\usepackage{graphicx}
\usepackage{subcaption}
\usepackage{url}
\usepackage{hyperref}
\hypersetup{colorlinks=true}
\hypersetup{linkcolor=black,anchorcolor=black,citecolor=black,urlcolor=black}
\hypersetup{pdftitle=Weighted k-Server Bounds via Combinatorial Dichotomies}
\hypersetup{pdfauthor={Nikhil Bansal} {Marek Elias} {Grigorios Koumoutsos}}
\hypersetup{pdfkeywords={weighted k-server problem} {online algorithms}
	{competitive analysis}}
\usepackage[algo2e,ruled]{algorithm2e}

\newtheorem{theorem}{Theorem}[section]
\newtheorem{definition}[theorem]{Definition}

\newtheorem{observation}[theorem]{Observation}
\newtheorem{lemma}[theorem]{Lemma}

\newcounter{note}[section]

\DeclareMathOperator{\ALG}{ALG}
\DeclareMathOperator{\ADV}{ADV}
\DeclareMathOperator{\OPT}{OPT}
\DeclareMathOperator{\WF}{WF}
\DeclareMathOperator{\WFA}{WFA}

\DeclareMathOperator{\SWF}{SW}
\DeclareMathOperator{\cost}{cost}

\DeclareMathOperator{\cl}{cl}
\newcommand{\cI}{\ensuremath{\mathcal{I}}}
\newcommand{\cL}{\ensuremath{\mathcal{L}}}
\newcommand{\cS}{\ensuremath{\mathcal{S}}}

\newcommand{\cQ}{\ensuremath{\mathcal{Q}}}
\newcommand{\cE}{\ensuremath{\mathcal{E}}}
\newcommand{\cG}{\ensuremath{\mathcal{G}}}
\DeclareMathOperator*{\argmin}{arg\,min}

\begin{document}

\maketitle

\begin{abstract}

The weighted $k$-server problem is a natural generalization of the $k$-server problem where each server has a different weight. We consider the problem on uniform metrics, which corresponds to a natural generalization of paging. Our main result is a doubly exponential lower bound on the competitive ratio of any deterministic online algorithm, that essentially matches the known upper bounds for the problem and closes  a large and long-standing gap.

The lower bound is based on relating the weighted $k$-server problem to a certain combinatorial problem and proving a Ramsey-theoretic lower bound for it.
 This combinatorial connection also reveals several structural properties of low cost feasible solutions to serve a sequence of requests. We use this to show that the generalized Work Function Algorithm achieves an almost optimum competitive ratio, and to obtain new refined upper bounds on the competitive ratio for the case of $d$ different weight classes.
\end{abstract}

\newpage
\tableofcontents
\newpage

\section{Introduction}
The $k$-server problem is one of the most natural and fundamental online problems and its study has been quite influential in the development of competitive analysis (see e.g.~\cite{BEY98,KP95,Kou09,Sit14,BBMN15}).
The problem is almost settled in the deterministic case: no algorithm can be better than $k$-competitive in any metric space of more than $k$ points \cite{MMS90},
and in their breakthrough result, Koutsoupias and Papadimitriou \cite{KP95} showed that the Work Function Algorithm (WFA) is $(2k-1)$-competitive in any metric space. Tight $k$-competitive algorithms are also known for several special metrics \cite{ST85,CKPV91,CL91,KP96}.

Despite this progress, several natural variants and generalizations of the $k$-server problem are very poorly understood.
In particular, they exhibit very different and intriguing behavior and the techniques for the standard $k$-server problem do not seem to apply to them
(we describe some of these problems and results in Section \ref{sec:rel_work}). 
Getting a better understanding of such problems is a natural step towards building a deeper theory of online computation.

\paragraph{Weighted $k$-server.}
Perhaps the simplest such problem is the weighted $k$-server problem on uniform metrics, that was first introduced and studied by Fiat and Ricklin~\cite{FR94}. Here, there are $k$ servers located at points of a uniform metric space. In each step a request arrives at some point and must be served by moving some server there. Each server $s_i$ has a positive weight $w_i$ and it costs $w_i$ to move $s_i$ to another point. The goal is to minimize the total cost for serving the requests.

Note that in the unweighted case where each $w_i=1$, this is the classic and extensively studied paging/caching problem \cite{ST85}, for which several tight $k$-competitive deterministic and $O(\log k)$-competitive randomized algorithms are known \cite{BEY98}.
Indeed, one of the motivations of \cite{FR94} for studying the weighted $k$-server problem was that it corresponds to paging where each memory slot has a different replacement cost.\footnote{We crucially note that this problem should not be confused by the related, but very different, weighted paging problem where the weights are on the pages instead of the servers. Weighted paging corresponds to (unweighted) $k$-server on weighted star metrics and is very well understood. In particular, tight $k$-competitive deterministic and $O(\log k)$-competitive randomized algorithms are known~\cite{CKPV91,Young94,BBN12}.}

Throughout this paper, we only consider the uniform metric, and by weighted $k$-server we always mean the problem on the uniform metric, unless stated otherwise.

\paragraph{Previous Bounds.}
There is surprisingly huge gap between the known upper and lower bounds on the
competitive ratio for weighted $k$-server.
In their seminal paper, Fiat and Ricklin \cite{FR94} gave the first deterministic  algorithm with a doubly exponential competitive ratio of about $2^{4^k} = 2^{2^{2k}}$.
They also showed a (singly) exponential lower bound of $(k+1)!/2$ on the competitive ratio of deterministic algorithms, which can be improved to
$(k+1)!-1$ by a more careful argument \cite{ChV13}.

More recently, 
 Chiplunkar and Vishwanathan \cite{ChV13} considered a simple memoryless randomized algorithm, where server $s_i$ moves to the requested point with some fixed probability $p_i$. They showed that there is always a choice of $p_i$ as function of the weights, for which this gives an $\alpha_k < 1.6^{2^k}$-competitive algorithm against adaptive online adversaries. Note that
$\alpha_k \in [2^{2^{k-1}}, 2^{2^k}]$. They also showed that this ratio is tight for such randomized memoryless algorithms.
By the simulation technique of Ben-David et al.~\cite{BBKTW94} that relates different adversary models, this gives an implicit  $\alpha_k^2 \leq 2^{2^{k+1}}$-competitive deterministic algorithm\footnote{A more careful analysis shows that the Fiat-Ricklin algorithm \cite{FR94} is also $2^{2^{k+O(1)}}$ competitive \cite{Chip-pc}.}.

\paragraph{Conjectured upper bound.}
Prior to our work, it was widely believed that the right competitive ratio should be $(k+1)!-1$. In fact, \cite{ChV13} mention that WFA is a natural candidate to achieve this.

There are several compelling reasons for believing this. First, for classical $k$-server, the lower bound of $k$ is achieved in metric spaces with $n=k+1$ points, where each request is at the (unique) point with no online server. The $(k+1)!-1$ lower bound for weighted $k$-server~\cite{FR94,ChV13} also uses $n=k+1$ points. More importantly, this is in fact the right bound for $n=k+1$. This follows as the weighted $k$-server problem on $n$ points is a Metrical Service System (MSS)\footnote{This is a Metrical Task System \cite{BLS92} where the cost in each state is either $0$ or infinite (called \textit{forcing task systems} in \cite{MMS90}).} with $N=\binom{n}{k}k!$ states, which correspond to the $k$-tuples describing the configuration of the servers. It is known that WFA is $(N-1)$-competitive for any MSS with $N$ states \cite{CL96}. As $N=(k+1)!$ for $n=k+1$, this gives the $(k+1)!-1$ upper bound. Moreover, Chrobak and Sgall~\cite{CS04} showed that WFA is exactly $(k+1)!-1 =3! -1 =5$-competitive for $k=2$ servers (with arbitrary $n$), providing strong coincidental evidence for the $(k+1)!-1$ bound for general $k$.

\subsection{Our Results}
In this paper, we study the weighted $k$-server problem systematically and
obtain several new results.
A key idea is to relate online weighted $k$-server to a natural {\em
offline} combinatorial question about the structure of all possible ``feasible
labelings'' for a hierarchical collection of intervals of depth $k$.
In particular, we show that the competitive ratio for weighted $k$-server is
closely related to a certain Ramsey-theoretic parameter of this combinatorial
problem. This parameter, let us call it $f(k)$ for the discussion here, 
reflects the amount of uncertainty that
adversary can create about the truly good solutions in an instance.
This connection is used for both upper and lower bound results in this paper.

\paragraph{Lower Bounds.} Somewhat surprisingly, we show that the doubly exponential upper bounds \cite{FR94,ChV13} for the problem
are essentially the best possible (up to lower order terms in the exponent).

\begin{theorem}
\label{thm:lb}
Any deterministic algorithm for the weighted $k$-server problem on uniform metrics has a competitive ratio at least $\Omega(2^{2^{k-4}})$.
\end{theorem}

As usual, we prove Theorem \ref{thm:lb} by designing an adversarial strategy to produce an online request sequence dynamically (depending on the actions of the online algorithm), so that
(i) the online algorithm incurs a high cost, while (ii) the adversary can always guarantee some low cost offline solution in hindsight.
Our strategy is based on a recursive construction on $n \geq \exp(\exp(k))$
points (necessarily so, by the connection to MSS) and it is designed in a modular way using the combinatorial connection as follows:
First, we construct a recursive lower bound instance for the combinatorial
problem for which the Ramsey-theoretic parameter $f(k) \geq 2^{2^{k-4}}$.
Second, to obtain the online lower bound, we embed this construction into a recursive strategy to dynamically generate an adversarial request sequence with the properties described above.

Moreover, we show that the lower bound from Theorem~\ref{thm:lb}, can be
extended to general metric spaces. That means, in any metric space containing
enough points, the competitive ratio of deterministic algorithms for weighted $k$-server is at least $\Omega(2^{2^{k-4}})$. We describe the details in Appendix~\ref{sec:general_lb}.

\paragraph{Upper Bounds.} The combinatorial connection is also very useful for positive results.
We first show that the generalized WFA, a very generic online algorithm that is applicable to a wide variety of problems, is essentially optimum.

\begin{theorem}\label{thm:ub}
The generalized $\WFA$ is $2^{2^{k+O(\log k)}}$-competitive for weighted
$k$-server on uniform metrics.
\end{theorem}
While previous algorithms \cite{FR94,ChV13} were also essentially optimum, this result is interesting as the generalized WFA is a generic algorithm and is not specifically designed for this problem at hand. In fact, as we discuss in Section~\ref{sec:rel_work}, for more general variants of $k$-server the generalized WFA seems to be only known candidate algorithm that can be competitive.

To show Theorem~\ref{thm:ub}, we first prove an almost matching upper bound of
$f(k) \leq 2^{2^{k+3\log k}}$ for the combinatorial problem. As will be clear later, we call such results {\em dichotomy theorems}.
Second, we relate the combinatorial problem to the dynamics of \textit{work functions} and use the dichotomy theorem recursively to bound the cost of the WFA on any instance.

This approach also allows us to extend and refine these results to the setting of $d$ different weight classes with $k_1,\ldots,k_d$ servers of each class. This corresponds to $d$-level caching where each cache has replacement cost $w_i$ and capacity $k_i$. As practical applications usually have few weight classes, the case where $d$ is a small constant independent of $k$ is of interest.
Previously, \cite{FR94} gave an improved $k^{O(k)}$ bound for $d=2$, but a major difficulty in extending their result is that their algorithm is phase-based and gets substantially more complicated for $d>2$.

\begin{theorem}\label{thm:ubd}
The competitive ratio of the generalized $\WFA$ for the weighted
$k$-server problem on uniform metrics with $d$ different weights is at most
$2^{O(d)\,k^3 \prod_{i=1}^d (k_i+1)}$, where $k_i$ is the number of servers
of weight $w_i$, and $k = \sum_{i=1}^d k_i$.
\end{theorem}

For $k$ distinct weights, i.e~$k_i=1$ for each $i$, note that this matches the $2^{\textrm{poly}(k) \cdot 2^k}$ bound in Theorem \ref{thm:ub}.
For $d$ weight classes, this gives $2^{O(dk^{d+3})}$, which is
singly exponential in $k$ for $d=O(1)$.
To prove Theorem \ref{thm:ubd}, we proceed as before. We first prove a more refined dichotomy theorem (Theorem \ref{thm:dichotomy-d}) and use it recursively with the WFA.

\subsection{Generalizations of $k$-server and Related Work}
\label{sec:rel_work}

The weighted $k$-server problem on uniform metrics that we consider here is the simplest among the several generalizations of $k$-server that
are very poorly understood. An immediate generalization is the weighted $k$-server problem in general metrics. This seems very intriguing even for a line metric.
Koutsoupias and Taylor \cite{KT04} showed that natural generalizations of many successful $k$-server algorithms are not competitive.
Chrobak and Sgall \cite{CS04} showed that any memoryless randomized algorithm has
unbounded competitive ratio. In fact, the only candidate competitive algorithm
for the line seems to be the generalized WFA.
There are also other qualitative differences. While the standard $k$-server problem is believed to have the same competitive ratio in every metric, this is not the case for weighted $k$-server.
For $k=2$ in a line, \cite{KT04} showed that any deterministic algorithm
is at least $10.12$-competitive, while on uniform metrics the competitive
ratio is 5 \cite{CS04}.

A far reaching generalization of the weighted $k$-server problem is the {\em generalized} $k$-server problem \cite{KT04, SS06,SSP03,Sit14}, with various applications.
Here, there are $k$ metric spaces $M_1,\ldots,M_k$, and each server $s_i$ moves in its own space $M_i$. A request $r_t$ at time $t$ is specified by a $k$-tuple
$r_t = (r_t(1),\ldots,r_t(k))$ and must be served by moving server $s_i$ to
$r_t(i)$ for some $i \in [k]$. Note that the usual $k$-server corresponds to
very special case where the metrics $M_i$ are identical and each request, $r_t =(\sigma_t,\sigma_t,\ldots,\sigma_t)$, has all coordinates identical. Weighted $k$-server (in a general metric $M$) is also a very special case where each $M_i = w_i \cdot M$ and $r_t =(\sigma_t,\sigma_t,\ldots,\sigma_t)$.

In a breakthrough result, Sitters and Stougie \cite{SS06} gave an $O(1)$-competitive algorithm for the generalized $k$-server problem for $k=2$. Recently,
Sitters~\cite{Sit14} showed that the generalized WFA is also $O(1)$-competitive for $k=2$. Finding any competitive algorithm for $k>2$ is a major open problem, even in very restricted cases.
For example, the special case where each $M_i$ is a line, also called the  CNN problem, has received a lot of attention (\cite{KT04,Chr03,AG10,iw01,IY04}), but even here no competitive algorithm  is known for $k>2$.

\subsection{Notation and Preliminaries}\label{sec:prelim}
We now give some necessary notation and basic concepts, that will be crucial for
the technical overview of our results and techniques in Section \ref{sec:overview}.

\paragraph{Problem definition.} Let $M = (U,d)$ be a uniform metric space, where $U=\{1,\ldots,n\}$ is the set of points (we sometimes call them pages) and $d: U^2 \rightarrow \mathcal{R}$ is the distance function which satisfies $d(p,q) = 1$ for $p\neq q$, and $d(p,p) = 0$. There are $k$ servers $s_1,\dotsc,s_k$ with weights
$w_1 \leq w_2 \leq \dotsc \leq w_k$ located at points of $M$. The cost of moving server $s_i$ from the
point $p$ to $q$ is $w_i\cdot d(p,q) = w_i$. The input is a request sequence $\sigma = \sigma_1,\sigma_2,\ldots,\sigma_T$, where $\sigma_t \in U$ is the point requested at time $t$. At each time $t$, an online algorithm needs to have a server at $\sigma_t$, without the knowledge of future requests. The goal is to minimize the total cost for serving $\sigma$.

We think of $n$ and $T$ as arbitrarily large compared to $k$.  Note that if the weights are equal or similar, we can
use the results for the (unweighted) $k$-server problem with no or small loss, so $w_{\max}/w_{\min}$ should be thought of as arbitrarily large.
Also, if two weights are similar, we can treat them as same without much loss, and so in general it is useful to think of the weights as well-separated, i.e.~$w_i \gg w_{i-1}$ for each $i$.

\paragraph{Work Functions and the Work Function Algorithm.}
We call a map $C \colon \{1, \dotsc, k\} \to U$, specifying that server $s_i$ is at point $C(i)$, a {\em configuration} $C$.
Given a request sequence $\sigma = \sigma_1, \dotsc, \sigma_t$, let $\WF_t(C)$ denote the optimal cost to serve requests $\sigma_1, \dotsc, \sigma_t$ and end up in configuration $C$. The function $\WF_t$ is called \textit{work function} at time $t$.
Note that if the request sequence would terminate at time $t$, then $\min_C \WF_t(C)$ would be the offline optimum cost.

The Work Function Algorithm (WFA) works as follows: Let $C_{t-1}$  denote its configuration at time $t-1$. Then upon the request $\sigma_t$, WFA moves to the configuration $C$ that minimizes $\WF_t(C) + d(C,C_{t-1})$. Note that in our setting,  $d(C,C') = \sum_{i=1}^k w_i \mathbf{1}_{(C(i) \neq C'(i))}$.
Roughly, WFA tries to mimic the offline optimum while also controlling its movement costs. For more background on WFA, see~\cite{BEY98,CL96,KP95}.

The generalized Work Function Algorithm ($\WFA_{\lambda}$) is parameterized by a constant $\lambda \in (0,1]$, and at time $t$ moves to the configuration
$C_t = \argmin_{C} \WF_t(C) + \lambda d(C,C_{t-1}).$
For more on $\WFA_{\lambda}$, see~\cite{Sit14}.

\paragraph{Service Patterns and Feasible Labelings.}

We can view any solution to the weighted $k$-server problem as an interval
covering in a natural way.
For each server $s_i$ we define a set of intervals $\cI_i$ which
captures the movements of $s_i$ as follows: Let $t_1 < t_2 < t_3 <
\dotsb$ be the times when $s_i$ moves. For each move
at time $t_j$ we have an interval $[t_{j-1}, t_{j}) \in \cI_i$, which means that $s_i$
stayed at the same location during this time period.
We assume that $t_0=0$ and also add a final interval $[t_{\textrm{last}}, T+1)$, where $t_{\textrm{last}}$ is the last time
when server $s_i$ moved. So if $s_i$ does not move at all, $\cI_i$ contains the single interval $[0,T+1)$.
This gives a natural bijection between the moves of $s_i$ and the intervals in
$\cI_i$, and the cost of the solution equals to $\sum_{i=1}^k w_i(|\cI_i|-1)$. We call $\cI_i$ the $i$th level of intervals, and an interval in $I \in \cI_i$ a level-$i$ interval, or simply an $i$th level interval.

\begin{definition}[Service Pattern] We call the collection $\cI = \cI_1 \cup \dotsb\cup \cI_k$ a service pattern
if each $\cI_i$ is a partition of $[0, T+1)$ into half-open intervals.
\end{definition}
Figure \ref{fig:interval} contains an
example of such service pattern.

To describe a solution for a weighted
$k$-server instance completely, we label each interval $I\in \cI_i$ with a point
where $s_i$ was located during time period $I$. We can also decide to give no
label to some interval (which means that we don't care on which point the server is located). We call this a {\em labeled} service pattern.

\begin{definition}[Feasible Labeling] Given a service pattern $\cI$ and a request sequence $\sigma$, we say that a (partial) labeling $\alpha\colon \cI \to U$ is {\em feasible} with respect to $\sigma$, if for each time
$t\geq 0$ there exists an interval $I \in \cI$ which contains $t$ and
$\alpha(I) = \sigma_t$.
\end{definition}

We call a service pattern $\cI$ feasible with respect to $\sigma$
if there is some labeling $\alpha$ of $\cI$ that is feasible with respect to $\sigma$.
Thus the offline weighted $k$-server problem for request sequence $\sigma$ is equivalent to the problem of
finding the cheapest feasible service pattern for $\sigma$.

Note that for a fixed service pattern $\cI$, there may not exist any feasible labeling, or alternately there might exist many feasible labelings for it. Understanding the structure of the various possible feasible labelings for a given service pattern will play a major role in our results.

\begin{figure}[t!]
\hfill\includegraphics{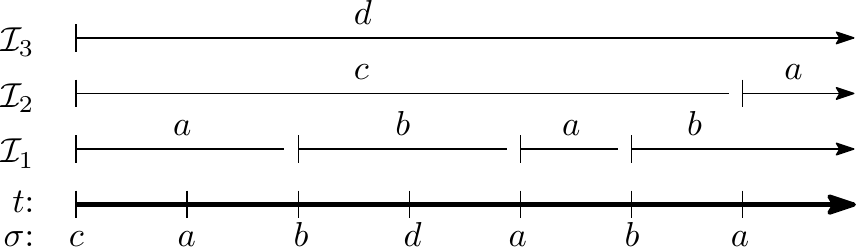}\hfill\ 
\caption{Illustration of a feasible service pattern for $k=3$. Each interval in
$\cI_i$ defines a location for server $s_i$. At each time $t$, some interval
covering $t$ should be labeled by the requested point $\sigma_t$.}
\label{fig:interval}
\end{figure}

\section{Overview}
\label{sec:overview}
We now give an overview of the technical ideas and the organization of the
paper.

Fix some request sequence $\sigma$.
Suppose that the online algorithm  knows the service pattern $\cI$ of some optimal offline solution, but not the actual labels for the intervals in $\cI$. Then intuitively, the online algorithm may still need to try out all possible candidate labels for an interval before figuring out the right one used by the offline solution.
So, a key natural question turns out to be: How does the structure of all possible
feasible labelings for $\cI$ look like?

Let us consider this more closely.
First, we can assume for simplicity that $\cI$ has a tree structure (i.e whenever an interval at level $i$ ends, all intervals at levels $1,\ldots,i-1$ end as well). Now, we can view $\cI$
as a collection of disjoint trees on different parts of $\sigma$, that do not interact with each other.
Focusing on some tree $T$ with a root interval $I$ (corresponding to the
heaviest server $s_k$), it now suffices to understand what is the number of
labels for $I$ in all feasible labelings with respect to $\sigma$.
This is because whenever we fix some label $a$ for $I$, we get a similar
question
about the depth-$(k-1)$ subtrees of $T$ on $\sigma$ with $a$ removed,
and we can proceed recursively.
This leads to the following problem.

\paragraph{The Combinatorial Problem.} Given an interval tree $T$ in a service
pattern $\cI$ on some request sequence $\sigma$. How many labels can the root
interval $I$ get over all possible feasible assignments to $\cI$?

We will show the following dichotomy result for this problem: (i) Either any label in $U$ works (i.e.~the location of $s_k$ does not matter), or (ii) there can be at most $f(k)$ feasible labels for $I$.
This might be somewhat surprising as the tree $T$ can be arbitrarily large
and the number of its subtrees of depth $k-1$ may not be even a function of $k$.

We prove two such dichotomy theorems in Section \ref{sec:dichotomy}. In Theorem
\ref{thm:dichotomy}, we show $f(k) = O(\exp(\exp(k))$ for arbitrary weights, and
in Theorem \ref{thm:dichotomy-d}, we give a more refined bound for the case
with $k_1,\ldots,k_d$ servers of weights $w_1,\ldots,w_d$.
These results are proved by induction but require some subtle technical details.
In particular, we need a stronger inductive hypothesis, where we track all the
feasible labels for the intervals on the path from the particular node towards
the root.
To this end, in Section \ref{sec:intervals} we describe some properties of
these path labelings and their interactions at different nodes.

\paragraph{Upper Bounds.}
These dichotomy theorems are useful to upper-bound the competitive ratio
as follows.
Suppose that the online algorithm knows the optimum service pattern $\cI$, but
not the actual labels.
Fix some tree $T \subseteq \cI$ with root interval $I$. We know that the offline
solution pays $w_k$ to move the server $s_k$ at the end of $I$,
and let  $\cost(k-1)$ denote the offline cost incurred during $I$ due to the
movement of the $k-1$ lighter servers.
Then, intuitively, the online algorithm has only $f(k)$ reasonable locations
\footnote{The situation in case (i) of the dichotomy theorem, where the location of $s_k$ does not matter, is much easier as the online algorithm can keep $s_k$ any location.}
to try during $I$.
Assuming recursively that its competitive ratio with $k-1$ servers
is $c_{k-1}$, its cost on $I$ should be at most
\[ f(k) \cdot (w_k + c_{k-1} \cdot \cost(k-1))
	\leq f(k) \cdot c_{k-1} (w_k + \cost(k-1))
	= f(k) \cdot c_{k-1} \cdot \OPT(I),
\]
which gives $c_k \leq f(k) c_{k-1}$, and hence $c_k \leq f(k) \cdots f(1)$.

Of course, the online algorithm does not know the offline service pattern $\cI$,
but we can remove this assumption by losing another factor $f(k)$.
The idea is roughly the following.
Consider some time period $[t_1,t_2]$, during which online incurs cost about
$c_{k-1}\,w_k$ and decides to move its heaviest server at time $t_2$.
We claim that there can be at most $f(k)$ locations for the heavy server
where the offline solution would pay less than $w_k/(4f(k))$ during $[t_1,t_2]$.
Indeed, suppose there were $m=f(k)+1$ such locations $p_1,\ldots,p_m$.
Then, for each $j=1,\ldots,m$,
take the corresponding optimum service pattern $\cI^j$ with $s_k$
located at $p_j$ throughout $[t_1,t_2]$, and consider a new pattern
$\cI'$ by taking the common refinement of $\cI^1,\ldots,\cI^m$
(where any interval in $\cI^j$ is a union of consecutive intervals in $\cI'$).
The pattern  $\cI'$ is quite cheap, its cost is at most
$m \cdot w_k/(4 f(k)) \leq w_k/2$, and we know that its root interval $I$ can
have $m=f(k)+1$ different labels. However, the dichotomy theorem implies that
any point is a feasible label for $I$, including the location of
the algorithm's heaviest server. But in such case, algorithm would not pay more
than $c_{k-1} \cost(\cI')$, what leads to a contradiction.

We make this intuition precise in Section \ref{sec:ubs} using work functions.
In particular, we use the idea above to show that
during any request sequence when $\WFA_{\lambda}$
moves $s_k$ about $f(k)$ times, any offline algorithm must pay $\Omega(w_k)$.

\paragraph*{Lower bound.}
In a more surprising direction, we can also use the combinatorial problem to
create a lower bound.
In Section \ref{sec:intervals}, we give a recursive combinatorial construction
of a request sequence $\sigma$ and a service pattern $\cI$ consisting of a
single interval tree, such that the number of feasible labelings for its root
can actually be about $r_k=2^{2^k}$.

Then in Section \ref{sec:lb}, we use the underlying combinatorial structure of this construction to
design an adversarial
strategy that forces any online algorithm to have a doubly-exponential
competitive ratio.
Our adversarial strategy reacts adaptively to the movements of the online
algorithm $\ALG$, enforcing the two key properties.
First, the adversary never moves a server $s_i$, where $i<k$,
unless $\ALG$ also moves some heavier server of weight at least $w_{i+1}$.
Second, the adversary never moves the heaviest server $s_k$ unless
$\ALG$ already moved $s_k$ to all $r_k$ possible feasible locations.
By choosing the weights of the servers well-separated, e.g.
$w_{i+1} \geq r_k \cdot \sum_{j=1}^i w_j$ for each $i$,
it is easy to see that the above two properties imply an $\Omega(r_k)$
lower bound on the competitive ratio.

\section{Service Patterns}
\label{sec:intervals}
In this section, we study the structure of feasible labelings. A crucial notion for this will be {\em request lists} of intervals.
We also define two Ramsey-theoretic parameters to describe the size of the request lists.
In Subsection~\ref{sec:int_lb}, we present a combinatorial lower bound
for these two parameters.

\paragraph{Hierarchical service patterns.}
We call a service pattern $\cI$ hierarchical, if each interval $I$ at level
$i < k$ has a unique parent $J$ at level $i+1$ such that $I \subseteq J$.
An arbitrary service pattern $\cI$ can be made hierarchical easily and at relatively
small cost: whenever an interval at level $i>1$ ends at time $t$,
we also end all intervals at levels $j=1, \dotsc, i-1$.
This operation adds at most $w_1 + \ldots w_{i-1} \leq  kw_i$, for each interval of weight $w_i$, so the overall cost
can increase by a factor at most $k$.
In fact, if the weights are well-separated, the loss is even smaller.

Henceforth, by service pattern we will always mean hierarchical service patterns, which we view as a disjoint collection of trees.
We adopt the usual terminology for trees. The
ancestors of $I$ are all intervals at higher levels containing $I$, and
descendants of $I$ are all intervals at lower levels which are contained in $I$.
We denote $A(I)$ the set of the ancestors of $I$ and $T_I$ the subtree of
intervals rooted at $I$ (note that $T_I$ includes $I$).

\paragraph{Composition of feasible labelings.}
In hierarchical service patterns, the labelings can be composed easily in
modular way.
Let $\sigma_I$ be the request sequence during the time
interval $I$, and $\sigma_J$ during some sibling $J$ of $I$.
If $\alpha_I$ and $\alpha_J$ are two feasible labelings with respect to
$\sigma_I$ and $\sigma_J$ respectively, and if they assign the same labels to
the ancestors of $I$ and $J$ (i.e.~$\alpha_I(A(I)) = \alpha_J(A(J))$),
we can easily construct a single $\alpha$ which is feasible with respect to both
$\sigma_I$ and $\sigma_J$:
Label the intervals in $T_I$ according to $\alpha_I$,
intervals in $T_J$ according to $\alpha_J$ and their ancestors according to
either $\alpha_I$ or $\alpha_J$.

\subsection{Structure of the feasible labelings}\label{sec:labelings}

Consider a fixed service pattern $\cI$ and some request sequence $\sigma$.
There is a natural inductive approach for understanding the structure of feasible labelings of $\cI$.
Consider an interval $I \in \cI$ at level $\ell < k$, and the associated request sequence $\sigma(I)$.
In any feasible labeling of $\cI$, some requests in $\sigma(I)$ will be covered (served) by the labels for the intervals
in  $T_I$, while others (possibly none) will be covered by labels assigned to ancestors $A(I)$ of $I$.
So, it is useful to understand how many different ``label sets" can arise for $A(I)$ in all possible feasible labelings.
This leads to the notion of request lists.

\paragraph{Request lists.}
Let $I$ be an interval at level $\ell<k$. We call a set of pages $S \subseteq U$ with $|S|\leq k-\ell$, a {\em valid tuple} for $I$, if upon assigning $S$ to ancestors of $I$ (in any order) there is some labeling of $T_I$ that is feasible for $\sigma_I$. Let $R(I)$ denote the collection of all valid tuples for $I$.

Note that if $S$ is a valid tuple for $I$, then all its supersets of
size up to $k-\ell$ are also valid. This makes the set of all valid tuples hard to work with, and so we only consider the inclusion-wise
minimal tuples.

\begin{definition}[Request list of an interval]
Let $I$ be an interval at level $\ell < k$.
The request list of $I$, denoted by $L(I)$, is the set of inclusion-wise minimal
valid tuples.
\end{definition}

{\em Remark.} We call this a request list as we view $I$ as requesting a tuple in $L(I)$ as ``help" from its ancestors in $A(I)$, to feasibly cover $\sigma(I)$. It is possible that there is a labeling $\alpha$ of $T_I$  that can already cover $\sigma(I)$ and hence $I$ does not
need any ``help" from its ancestors $A(I)$. In this case $L(I)=\{\emptyset\}$
(or equivalently every subset $S \subset U$ of size $\leq k-\ell$ is a valid tuple).

Tuples of size $1$ in a request list will play an important role, and we will call them singletons.

{\bf Example.}
Let $I\in \cI$ be an interval at level $1$, and $I_2,
\dotsc, I_k$ its ancestors at levels $2, \dotsc, k$.
If $P=\{p_1,\ldots,p_j\}$, where $j<k$, is the set of all pages requested in  $\sigma_I$, then one feasible labeling $\alpha$ with respect to $\sigma_I$ is  to
assign
 $\alpha(I_{i+1}) = p_i$ for $i=1, \dotsc, j$,
and no label for any other $J\in \cI$. So $P$ is a feasible tuple.
However, $P$ is not inclusion-wise minimal,
as $\{p_2, \dotsc, p_j\}$ is also valid tuple:
We can set $\alpha(I) = p_1$, $\alpha(I_i) = p_i$ for $i=2,\dotsc, j$ and
no label for other intervals.
Similarly, $P\setminus \{p_i\}$ for $i=2, \dotsc, j$, are also valid and inclusion-wise minimal.
So, we have
\[ L(I) = \big\{ P\setminus \{p_1\}, P\setminus \{p_2\}, \ldots,
P\setminus\{p_j\} \big\}. \]

\paragraph{Computation of Request Lists.}
Given a service pattern $\cI$ and request sequence $\sigma$, the request lists for each interval $I$ can be computed inductively.
For the base case of a leaf interval $I$ we already saw that $L(I) = \big\{ P\setminus \{p_1\}, P\setminus \{p_2\}, \ldots,
P\setminus\{p_j\} \big\}.$

For a higher level interval $I$, we will take the request lists of the children of $I$ and combine them suitably.
To describe this, we introduce the notion of  {\em joint request lists}.
Let $I$ be a level $\ell$ interval for $\ell>1$, and
let $C = \{J_1, \dotsc, J_m\}\subseteq \cI_{\ell-1}$ be the set of its child
intervals.
Note that $m$ can be arbitrarily large (and need not be a function of $k$).
We define the joint request list of the intervals in $C$ as follows.

\begin{definition}[Joint request list]
Let $I$ be an interval at level $\ell > 1$ and $C$ be the set of its children at
level $\ell-1$. The joint request list of $C$, denoted by $L(C)$,
is the set of inclusion-wise minimal tuples $S$ with $|S| \leq k-(\ell-1)$ for which
there is a labeling $\alpha$ that is feasible with respect to $\sigma_{I}$ and $ \alpha(\{I\}\cup A(I)) = S$.
\end{definition}

Let $R(C)$ denote collection of all valid tuples (not necessarily minimal) in the joint request list for $C$.
We note the following simple observation.

\begin{observation}\label{obs:list_prod}
A tuple $S$ belongs to $R(C)$ if and only if $S$ belongs to $R(J_i)$ for each
$i=1, \dotsc, m$. This implies that $S \in L(C)$ whenever it is an
inclusion-wise minimal tuple such that each $L(J_i)$ for $i \in [m]$ contains some tuple $S_i \subseteq S$.
\end{observation}
\begin{proof} Consider the feasible labeling $\alpha$ with respect
to $\sigma_I$, which certifies that $S \in R(C)$. This is also feasible with respect
to each $\sigma_{J_i}$ and certifies that $S \in R(J_i)$.
Conversely, let $\alpha_i$, for $i=1,\dotsc, m$, denote the labeling
feasible with respect to $\sigma_{J_i}$ which certifies that $S\in R(C)$.
As $J_1, \dotsc, J_m$ are siblings, the composed labeling $\alpha$ defined by
 $\alpha(J) = \alpha_i(J)$ if $J\in T_{J_i}$ and, say, $\alpha(J) = \alpha_1(J)$
if $J \in  I \cup A(I)$ is feasible for $\sigma_I$.
\end{proof}

Creation of the joint request list can be also seen as a kind of a product
operation. For example, if there are two siblings $J_1$ and
$J_2$ whose request lists are disjoint and contain only singletons, then their
joint request list $L(J_1, J_2)$ contains all pairs $\{p,q\}$ such that
$\{p\}\in L(J_1)$ and $\{q\}\in L(J_2)$.
By Observation~\ref{obs:list_prod}, all such pairs belong to $R(J_1, J_2)$
and they are inclusion-wise minimal.
The number of pairs in $L(J_1, J_2)$ equals to $|L(J_1)|\cdot |L(J_2)|$.
In general, if $L(J_1)$ and $L(J_2)$ are not disjoint or contain tuples of different
sizes, the product operation becomes more complicated, and therefore we use
the view from Observation~\ref{obs:list_prod} throughout this paper.

Finally, having obtained $L(C)$, the request list $L(I)$ is obtained using the following observation.

\begin{observation}\label{obs:C_to_I_list}
A tuple $S$ belongs to $R(I)$ if and only if
$S\cup\{p\}$ belongs to $R(C)$ for some $p\in U$.
\end{observation}
\begin{proof}
If $S \cup \{p\} \in R(C)$, then we find a feasible labeling $\alpha$ for $\sigma_I$ with $\alpha(I)=p$ and  $\alpha(A(I))=S$. Conversely, if $S \in R(I)$, then there must some feasible labeling $\alpha$ for $\sigma_I$ with
$\alpha(A(I))=S$, and we simply take $p=\alpha(I)$.
\end{proof}

So $L(I)$ can be generated by taking $S\setminus \{p\}$ for each
$p\in S$ and $S\in L(C)$, and eliminating all resulting tuples that are not
inclusion-wise minimal.

\vspace{2mm}
{\bf Example.} Consider the interval $I$ having two children
$J_1$ and $J_2$. We saw that if $L(J_1)$ and $L(J_2)$ are disjoint
and both contain only singletons, their joint list $L(J_1, J_2)$ contains
$|L(J_1)|\cdot |L(J_2)|$ pairs. Then, according to
Observation~\ref{obs:C_to_I_list}, $L(I)$ contains
$|L(J_1)| + |L(J_2)|$ singletons.

This observation that composition of request lists with singletons give request lists with singletons will be useful
in the lower bound below.

\paragraph{Sizes of request lists.}
Now we define the Ramsey-theoretic parameters of the service patterns.
Let us denote $f(\ell, t)$ the maximum possible numbers of $t$-tuples in the
request list $L(I)$ of any interval $I$ at level $\ell$ from any service pattern $\cI$.
Similarly, we denote $n(\ell,t)$ the maximum possible number of $t$-tuples
contained in a joint request list $L(C)$,
where $C$ are the children of some $\ell$th level interval $I$.
The examples above show that $n(\ell,2)$ can be of order
$f^2(\ell,1)$, and $f(\ell+1,1)\geq 2 f(\ell,1)$.
In the following subsection we show that the $f(\ell,1)$ and $n(\ell,1)$ can grow doubly
exponentially with $\ell$.

\subsection{Doubly-exponential growth of the size of request lists}
\label{sec:int_lb}
In the following theorem we show a doubly exponential lower bound on $n(k,1)$,
which is the maximum number of singletons in a joint request list of children of
any $k$th level interval.
In particular, we construct a request sequence $\sigma$, and a service pattern
such that each level-$\ell$ interval has a request list of $\Omega(2^{2^{\ell - 3}})$ singletons.
This construction is the key ingredient of the lower bound
in Section~\ref{sec:lb}.

\begin{theorem}\label{thm:int_lb}
The numbers $n(\ell,1)$ and $f(\ell-1,1)$ grow doubly-exponentially with $\ell$.
More specifically, for level $k$ intervals we have
\[ n(k,1) \geq f(k-1,1) \geq 2^{2^{k - 4}}. \]
\end{theorem}
\begin{proof}
We construct a request sequence and a hierarchical service pattern with a single $k$th level interval, such that any interval $I$ at level $1 \leq \ell < k$ has a request list $L(I)$ consisting of $n_{\ell+1}$ singletons, where $n_2=2$ and

 \[ n_{i+1} = (\lfloor n_i/2\rfloor + 1)
	+ (\lfloor n_i/2\rfloor + 1) \lceil n_i/2 \rceil \geq (n_i/2)^2.
\]

Note that $n_2=2, n_3=4, n_4=9, \ldots$ and in general as $n_{i+1} \geq (n_i/2)^2$ it follows that for $\ell \geq 4$, we have $n_{\ell} \geq 2^{2^{\ell-4}+2} \geq 2^{2^{\ell-4}}$.

We describe our construction inductively, starting at the first level. Let $I$ be an interval at level $1$ and $\sigma_I$ be a request sequence consisting of $n_2=2$ distinct pages $p$ and $q$. Clearly, $L(I) = \big\{\{p\},\{q\}\big\}$. We denote the subtree at $I$ together with the request sequence $\sigma_I$ as $T_1(\{p,q\})$.

Now, for $i \geq 2$, let us assume that we already know how to construct a tree $T_{i-1}(P)$ which has a single $(i-1)$th level interval $J$, its request sequence $\sigma_J$ consists only of
pages in $P$, and for each $p\in P$, $\{p\}$ is contained as a singleton in $L(J)$. Let $n_i$ denote the size of $P$.

We show how to construct $T_{i}(P')$,
for an arbitrary set $P'$ of $n_{i+1}$ pages, such that $T_i$ has a single $i$th level interval $I$, and all pages $p \in P'$ are contained in the request list $L(I)$ as singletons.

First, we create a set of pages $M \subset P'$ called {\em mask}, such that $|M| = \lfloor n_i/2\rfloor + 1$. Pages that belong to $M$ are arbitrarily chosen from $P'$. Then, we partition $P' \setminus M$ into $|M|$ disjoint sets of size $\lceil n_i/2\rceil$ and we associate each of these sets with a page $q \in M$. We denote by $Q_q$ the set associated to page $q \in M$. For each $q \in M$, let $P_q = (M\setminus \{q\}) \cup Q_q$. Note that $|P_q|=n_i$. The interval tree $T_{i}(P')$ is created as follows. It consists of a single interval $I$ at level $i$ with $\lfloor n_i/2\rfloor + 1$ children $J_q$. For each $J_q$ we inductively create a level $i-1$ subtree $T_{i-1}(P_q)$.
See Figure~\ref{fig:interval_lb} for an example.

\begin{figure}
\hfill\includegraphics{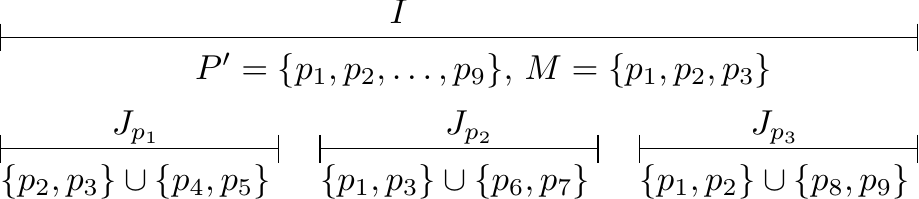}\hfill\ 
\caption{
Construction of an interval at level $i=3$ with request list of $n_4=9$ singletons, using intervals of level $i-1=2$ with request lists of $n_3 = 4$ singletons. The set $P'$ is decomposed into a mask of size $ \lfloor n_3/2\rfloor + 1  =3$, $M = \{ p_1,p_2,p_3 \}$ and sets $Q_{p_1} = \{ p_4,p_5 \}$, $Q_{p_2} = \{ p_6,p_7 \}$ and $Q_{p_3} = \{ p_8, p_9 \} $. For each $q \in M$, the requested set of interval $J_q$ is $P_q = (M \setminus q) \cup Q_q$.}
\label{fig:interval_lb}
\end{figure}

This construction has two important properties:
\begin{lemma}
\label{lem:2props}
First, for each $p\in P'$ there exists a subtree $T_{i-1}(P_q)$
	such that $p\notin P_q$.
Second, for each $p\in P'$ there exists a page $\bar{p}\in P'$ such that each $P_q$ contains either $p$ or $\bar{p}$.
\end{lemma}
\begin{proof}
If page $p\in M$, then it
belongs to all sets $P_q$ except for $P_p$.
If $p \notin M$, then $p \in Q_q$ for some $q$ and hence $p$ only lies in $P_q$.
 This proves the first property.

For the second property, if $p\in M$, we can choose an arbitrary
$\bar{p} \in P_p$, and note that $p$ lies in every $P_q$ for $q \in M \setminus\{p\}$.
 On the other hand, if $p \notin M$, let $q \in M$ be such that $p\in Q_q$ (and hence $p \in P_q$) and we define $\bar{p}=q$. Then by construction, $q$ is contained in all other sets $P_{q'}$ for $q'\neq M \setminus \{q\}$.
\end{proof}

Using the above lemma, we can now understand the structure of request lists.
\begin{lemma} The request list $L(I)$ consists of all singletons in $P'$, i.e.~$L(I)=\{ \{p\}\;|\, p\in P'\}$.
\end{lemma}
\begin{proof}
Let us assume by inductive hypothesis that for each child $J_q$ of $I$ we have
$L(J_q) = \{ \{p\}\;|\, p\in P_q\}$. As discussed above, this is true for the base case of
intervals at level $1$.

By Observation~\ref{obs:list_prod} and by the first property in Lemma \ref{lem:2props}, no singleton
belongs to $R(C)$. By the second property we also know that each $p\in P'$
appears in some pair $\{p, \bar{p}\}$ in $R(C)$.
Therefore, by Observation~\ref{obs:C_to_I_list},
we know that $R(I)$ contains a singleton
$\{p\}$ for each $p\in P'$, and also that $R(I)$ does not contain empty set,
since $R(C)$ contains no singletons. So, we have $ L(I) = \big\{ \{p\}\;|\, p\in P'\big\}.$
\end{proof}

This completes the construction of the tree $T_{k-1}(P)$ with $P$ of size $n_{k}$
and with a single interval $J$ at level $k-1$. To finish the service pattern, we
create a single $k$th level interval $I$ having $J$ as its only child.
By the discussion above,
$L(J)$ contains a singleton $\{p\}$ for each $p\in P$ and trivially
the joint request list of the children of $I$ is simply $L(J)$. Therefore we
have $ n(k,1)\geq f(k-1,1) \geq n_{k} \geq 2^{2^{k-4}}.$
\end{proof}
 
\section{Online Lower bound}
\label{sec:lb}
In this section we transform the combinatorial construction of
Theorem~\ref{thm:int_lb} into a lower bound for any deterministic algorithm,
proving Theorem~\ref{thm:lb}. Throughout this section we denote $s^{\ALG}_1,
\dotsc, s^{\ALG}_k$ the servers of the online algorithm and
$s^{\ADV}_1, \dotsc, s^{\ADV}_k$ the servers of the adversary.

Here is the main idea.
Let $\ALG$ be a fixed online algorithm. We create a request sequence adaptively, based on the decisions of $\ALG$, which consists of arbitrary number of \textit{phases}. During each phase, the heaviest server of the adversary $s^{\ADV}_k$ stays at some fixed location. Whenever a phase ends, the adversary might move all its servers (including $s^{\ADV}_k$) and a new phase may start. During a phase, requests are determined by a recursive construction, using \textit{strategies} which we define later on. At a high-level, the goal of the strategies is to make sure that the following two properties are satisfied:
\begin{enumerate}[(i)]
\item \label{prtA} For $i=1,\ldots,k-1$, $\ADV$ never moves server $s_i^{\ADV}$,
unless $\ALG$ moves some heavier server $s_j^{\ALG}$ for $j > i$ at the same
time.
\item \label{prtB} During each phase, $\ALG$ moves its heaviest server $s_k^{\ALG}$ at least $n_k$ times.
\end{enumerate}
 These two properties already imply a lower bound on the competitive ratio of
$\ALG$ of order $n_k$, whenever the weights of the servers are well separated, i.e.
 $w_{i+1} \geq n_k \cdot \sum_{j=1}^i w_j$ for each $1\leq i<k$. Here $n_k \geq 2^{2^{k-4}}$ is the number of candidate points for $s_k^{\ADV}$. 

In the following section we show how each phase is defined using strategies. We
conclude the proof of Theorem~\ref{thm:lb} in Section \ref{sec:lb_thm}.

\subsection{Definition of Strategies}
\label{sec:lb-str}

Each phase is created using $k$ adaptive \textit{strategies} $S_{1},\dotsc,S_k$, where $S_1$ is the simplest one and $S_{i+1}$ consists of several executions of $S_i$. An execution of strategy $S_i$ for $1 \leq i < k$ corresponds to a subsequence of requests where $\ALG$ moves only servers $s_1,\dotsc,s_{i}$. Whenever $\ALG$ moves some server $s_j$, for $j >i$, the execution of $S_i$ ends. An execution of $S_k$ ends only when $\ALG$ moves its heaviest server to the location of $s_k^{\ADV}$.

We denote by $S_i(P)$ an execution of strategy $S_i$, with a set of requested
points $P$, where $|P| = n_{i+1}$. We start by defining the strategy of the highest level $S_k$. An execution $S_k(P)$ defines a \textit{phase} of the request sequence. We make sure that if $p$ is the location of $s_k^{\ALG}$ when the execution starts, then $p \notin P$.

\begin{algorithm2e}
\SetAlgoRefName{}
\SetAlgorithmName{Strategy $S_k(P)$}{}{}
\caption{\ }
partition $P$ into $T$ and $B$ of size $n_k$ each arbitrarily\;
$B' := \emptyset$\;
\While{$T \neq \emptyset$}{	
	Run $S_{k-1}(T\cup B')$ until $\ALG$ moves $s_k$\;
	$p := \text{new position of $s_k^{\ALG}$}$\;
	$T := T \setminus \{p\}$\;
	$B' := \text{arbitrary subset of $B \setminus \{p\}$
		of size $n_{k}-|T|$}$\;
}
Terminate Phase
\end{algorithm2e}

Intuitively, we can think of $T$ as the set of candidate locations for $s_k^{\ADV}$.
The set $B'$ is just a padding of new pages to construct a set $T \cup B'$ of size
$n_k$ as an argument for $S_{k-1}$. Whenever $s_k^{\ALG}$ is placed
at some point $p \in T$, we remove it from $T$, otherwise $T$ does not change. We then update $B'$ such that $|T \cup B'| = n_k$ and $p \notin B'$. This way, we make sure that $p$ is never requested as long as $s_k^{\ALG}$ stays there.

We now define the strategies $S_i$ for $1 < i < k$. An execution of $S_i(P)$ executes several consecutive instances of $S_{i-1}$. We first describe how do we choose the set $P'$ for each execution of $S_{i-1}(P')$, such that $P' \subset P$ and $|P'| = n_i$.

We use the construction described in Section~\ref{sec:int_lb}. In particular, we choose an arbitrary set of points $M \subset P$ called {\em mask}, and we partition $P \setminus M$ into $|M|$ disjoint sets of equal size, each one associated with a point $q \in M$. We denote by $Q_q$ the set associated to point $q \in M$. All executions of strategy $S_{i-1}$, have as an argument a set $P_q = (M \setminus \{q\}) \cup Q_q $, for some $q \in M$. We discuss about the size of $M$ and $Q_q$ later on. Before describing the strategies, we observe that Lemma~\ref{lem:2props} implies the following two things regarding those sets:

\begin{observation}
For each point $p\in P$ there is a set $P_q$ such that $p \notin P_q$. If $p\in M$, this set is $P_p$. Otherwise, $p$ belongs to $Q_q$ for some $q\in M$ and then we can choose any $P_{q'}$ for $q'\ne q$.
\end{observation}

\begin{observation}\label{obs:lb_2pages}
For any $p\in P$ there is a point $\bar{p}$ such that each $P_q$ contains either $p$ or
$\bar{p}$. In particular, if $p\in Q_q$ for some $q$, we choose $\bar{p}=q$, otherwise $p\in M$ and then we can choose any $\bar{p} \in Q_p$. 
\end{observation}

\begin{algorithm2e}[H]
\SetAlgoRefName{}
\SetAlgorithmName{Strategy $S_i(P)$, where $1<i<k$}{}{}
\caption{\ }
Decompose $P$ into mask $M$ and sets $Q_q$\;
For each $q \in M$, denote $P_q := (M \setminus \{q\}) \cup Q_q$\;
\Repeat{$\ALG$ moves $s_{i+1}$ or some heavier server}{
	$p := \text{position of $s_i^{\ALG}$}$\;
	Choose any $P_q$, s.t. $p\notin P_q$,
	and	run $S_{i-1}(P_q)$ until $\ALG$ moves $s_i$\;
}
\end{algorithm2e}

Last, strategy $S_1$ takes as an argument a set of $n_2=2$ points and requests them in an alternating way.

\begin{algorithm2e}[H]
\SetAlgoRefName{}
\SetAlgorithmName{Strategy $S_1(\{p,q\})$}{}{}
\caption{\ }
\Repeat{$\ALG$ moves $s_2$ or some heavier server}{
	If $s_1^{\ALG}$ is at $q$: request $p$\;
	Otherwise: request $q$\;
}
\end{algorithm2e}

Observe that, an execution of a strategy $S_i$, for $1 \leq i < k$, ends only if $\ALG$ moves some heavier server. This means that if $\ALG$ decides not to move any heavier server, then the execution continues until the end of the request sequence. Moreover, it is crucial to mention that, by construction of the strategies $S_1, \dotsc, S_k$, we have the following:

\begin{observation}
For $1 \leq i \leq k$, if server $s_i^{\ALG}$ is located at some point $p$, then $p$
is never requested until $s_i^{\ALG}$ moves elsewhere.
\end{observation}

\paragraph{Cardinality of sets.} We now determine the size of the sets used by $S_i$ for $1 < i \leq k $. For $2 \leq i \leq k$, recall that all arguments of $S_{i-1}$ should have size $n_i$. In order to satisfy this, we choose the sizes as in Section~\ref{sec:int_lb}, i.e. $|M| = \lceil n_i/2\rceil + 1$ and $|Q_q| = \lfloor n_i/2 \rfloor $. It is clear that $ |P_q| = |(M\setminus \{q\})| + |Q_q|  = n_i  $. Recall that $P = M \cup (\bigcup_{q\in M} Q_q)$, and therefore we have

 \begin{equation*}
n_{i+1} = |P| = (\lceil n_i/2\rceil + 1)
	+ (\lceil n_i/2\rceil + 1)\lfloor n_i/2 \rfloor
	= (\lceil n_i/2\rceil + 1)(\lfloor n_i/2 \rfloor+1)
	\geq n_i^2/4.
\end{equation*}

Therefore, by choosing $n_2 = 2$ we have $n_3=4$ and for $k \geq 4$
\begin{equation}
\label{eq:nklb}
n_k \geq 2^{2^{k-4}+2}  \geq 2^{2^{k-4}}.
\end{equation}

Last, for strategy $S_k$ we have that $n_{k+1} = 2 n_k$.

\paragraph{Service pattern associated with the request sequence.} We associate
the request sequence with a service pattern $\mathcal{I}$,
which is constructed as follows: For each execution of strategy $S_i$, we create
one interval $I \in \mathcal{I}_i $. We define $\mathcal{I} = \mathcal{I}_1 \cup
\cdots \cup \mathcal{I}_k$. Clearly, $\mathcal{I}$ is hierarchical.
Next lemma gives a characterization of request lists of intervals $I \in \mathcal{I}$.

\begin{lemma}\label{lem:pages_useful}
Let $I$ be the interval corresponding to a particular instance $S_i(P)$.
Given that there is an interval $J$ at level $j > i$ such that $J$ is labeled
by some point $p\in P$, then there is a feasible labeling for $I$ and its
descendants covering all requests issued during the lifetime of $I$.

In other words, all points $p \in P$ are contained in the request list $L(I)$ as singletons.
\end{lemma}

\begin{proof}
We prove the lemma by induction.
The lemma holds trivially for
$S_1$:
The requests are issued only in two different points $p$ and $q$.
Whenever $J$ has assigned $p$, we assign $q$ to $I$ and vice versa. In both
cases all the requests during the lifetime of $I$ are covered either by $I$ or
by $J$.

Now, assuming that lemma holds for level $i-1$, let us prove it also for
$i$. Let $p \in P$ be the point assigned to $J$ the ancestor of $I$.
By Observation~\ref{obs:lb_2pages}, we know that there is a $\bar{p}\in P$
such that each $P_q$ contains either $p$ or $\bar{p}$. We assign $\bar{p}$ to
$I$ and this satisfies the condition of the inductive hypothesis for all
children of $I$, as all those instances have one of $P_q$ as an input.
\end{proof}

Moreover, by construction of strategy $S_k$, we get the following lemma.

\begin{lemma}
\label{lem:lb-feas}
The service pattern $\mathcal{I}$ associated with the request sequence is feasible.
\end{lemma}

\begin{proof}
Let $\cI_P$ the service pattern associated with an execution of $S_k(P)$. We show that $\cI_P$ is feasible. Since the request sequence consists of executions of $S_k$, the lemma then follows.
 Let $p$ be the last point which remained in $T$ during the execution of
$S_k(P)$. Then, all the former children of $S_k$ were of form $S_{k-1}(P')$ where
$p \in P'$. By Lemma~\ref{lem:pages_useful}, by assigning $p$ to the $k$th level interval, there is a feasible
labeling for each children interval and all their descendants. We get that the service pattern $\mathcal{\cI_P}$ associated with strategy $S_k(P)$ is feasible.
\end{proof}

\subsection{Proof of Theorem~\ref{thm:lb}}\label{sec:lb_thm}

We now prove Theorem~\ref{thm:lb}. We first define the moves of $\ADV$ and then we proceed to the final calculation of the lower bound on the competitive ratio of $\ALG$. Recall that the request sequence consists of arbitrary many phases, where each phase is an execution of strategy $S_k$.

\paragraph{Moves of the adversary.} Initially, we allow the adversary to move all its servers, to prepare for the first phase. The cost of this move is at most $\sum_{i=1}^{k} w_i $. Since the request sequence can be arbitrarily long, this additive term does not affect the competitive ratio and we ignore it for the rest of the proof. It remains to describe the moves of the adversary during the request sequence.

   Consider the service pattern $\mathcal{I}$ associated with the request sequence. By lemma \ref{lem:lb-feas}, $\mathcal{I}$ is feasible. We associate to the adversary a feasible assignment of $\mathcal{I}$. This way, moves of servers $s^{\ADV}_i$ for all $i$ are completely determined by $\mathcal{I}$. We get the following lemma.

\begin{lemma}
\label{lem:prta}
At each time $t$, $\ADV$ does not move server $s_i^{\ADV}$, for $i=1,\ldots,k-1$, unless $\ALG$ moves some heavier server $s_j^{\ALG}$ for $j > i$.
\end{lemma}

\begin{proof}
Consider the service pattern $\mathcal{I}$ associated with the request sequence. Each execution of strategy $S_i$ is associated with an interval $I_i \in \mathcal{I}$ at level $i$. The adversary moves $s_i^{\ADV}$ if and only if interval $I_i$ ends. By construction, interval $I_i$ ends if and only if its corresponding execution of $S_i$ ends. An execution of $S_i$ ends if and only if $\ALG$ moves some heavier server $s_j$ for $j >i$.
We get that at any time $t$, server $s_i^{\ADV}$ moves if and only if $\ALG$ moves some server $s_j$ for $j >i$.
\end{proof}

\paragraph{Calculation of the lower bound.} Let $\cost(s_i^{\ALG})$ and $\cost(s_i^{\ADV})$ denote the cost due to moves of the $i$th server of $\ALG$ and $\ADV$ respectively. Without loss of generality\footnote{It is easy to see that, if $\ALG$ uses only its lightest server $s_1^{\ALG}$, it is not competitive: the whole request sequence is an execution of $S_2(p,q)$, so the adversary can serve all requests at cost $w_2 + w_1$ by moving at the beginning $s_2^{\ADV}$ at $q$ and $s_1^{\ADV}$ at $p$. $\ALG$ pays 1 for each request, thus its cost equals the length of the request sequence, which implies an unbounded competitive ratio for $\ALG$.}, we assume that $ \sum_{i=2}^{k} \cost(s_i^{\ALG}) > 0 $. Recall that we assume a strong separation between weights of the servers. Namely, we have $w_1 = 1$ and $w_{i+1} = n_k \cdot \sum_{j=1}^i w_i$. This, combined with lemma (\ref{lem:prta}) implies that
\begin{equation}
\label{eq:prta}
\sum_{i=1}^{k-1} \cost(s_i^{\ADV}) \leq \big( \sum_{i=2}^{k} \cost(s_i^{\ALG}) \big)/n_k .
\end{equation}
Moreover, by construction of strategy $S_k$, a phase of the request sequence ends if and only if $\ALG$ has moved its heaviest server $s_k^{\ALG}$ at least $n_k$ times. For each phase $\ADV$ moves $s_k^{\ADV}$ only at the end of the phase. Thus we get that
\begin{equation}
\label{eq:prtb}
\cost(s_k^{\ADV}) \leq \cost(s_k^{\ALG})/n_k.
\end{equation}
Overall, using \eqref{eq:prta} and \eqref{eq:prtb} we get
\begin{align*}
\cost(\ADV) & =  \sum_{i=1}^{k-1} \cost(s_i^{\ADV}) + \cost(s_k^{\ADV})   \leq  \big( \sum_{i=2}^{k} \cost(s_i^{\ALG}) \big)/n_k + \cost(s_k^{\ALG})/ n_k  \\
                  & =   \big( \sum_{i=2}^{k-1} \cost(s_i^{\ALG}) + 2 \cost(s_k^{\ALG}) \big) /n_k  \leq 2 \cdot \cost(\ALG)/n_k.
\end{align*}
Therefore, the competitive ratio of $\ALG$ is at least $n_k/2$, which by \eqref{eq:nklb} is $\Omega(2^{2^{k-4}})$.
\hfill\qedsymbol

\section{Dichotomy theorems for service patterns}
\label{sec:dichotomy}
The theorems proved in this section are matching counterparts to
Theorem~\ref{thm:int_lb} ---
they provide an upper bound for the size of the request lists in a fixed
service pattern.
The first one, Theorem~\ref{thm:dichotomy}, shows that the parameter $n(k,1)$,
as defined in Section~\ref{sec:intervals}, is at most doubly exponential in $k$.
This bound is later used in Section~\ref{sec:ub} to prove an upper bound
for the case of the weighted
$k$-server problem where all the servers might have a different weight.

Moreover, we also consider a special case when there are only $d<k$ different
weights $w_1, \dotsc, w_d$. Then, for each $i=1, \dotsc, d$, we have
$k_i$ servers of weight $w_i$. This situation can be modeled using a service
pattern with only $d$ levels, where each interval at level $i$ can be labeled by
at most $k_i$ pages. For such service patterns, we can get a stronger upper
bound, which is singly exponential in $d$, $k^2$ and the product $\prod k_i$,
see Theorem~\ref{thm:dichotomy-d}. This theorem is later used in
Section~\ref{sec:ubd} to prove a performance guarantee for $\WFA$ in this
special setting.

\subsection{General setting of weighted $k$-server problem}
\label{sec:dichotomy_k}

Recall the definitions from Section~\ref{sec:intervals},
where we denote $n(k,1)$ the maximum possible number of singletons contained in
a joint request list of children of some $k$th level interval.

\begin{theorem}[Dichotomy theorem for $k$ different weights]
\label{thm:dichotomy}
Let $\cI$ be a service pattern of $k$ levels and $I\in \cI$ be an arbitrary interval
at level $k$. Let $Q\subseteq U$ be the set of feasible labels for $I$.
Then either $Q = U$, or
$|Q| \leq n(k,1)$ and $n(k,1)$ can be at most $2^{2^{k+3\log k}}$.
\end{theorem}

First, we need to extend slightly our definitions of $f(\ell,t)$ and $n(\ell,t)$
from Section~\ref{sec:intervals}.
Let $I$ be an arbitrary interval at level $\ell$ and $J_1, \dotsc, J_m$ be its
children at level $\ell-1$.
We define $f(\ell, t, P)$ to be the maximum possible number of $t$-tuples in the
request list $L(I)$ such that all those $t$-tuples contain some
predefined set $P$, and we define $f(\ell, t, h)$ as a maximum such number over
all sets $P$ with $|P|= h$ pages. For example, note that we have $f(\ell, t, t) = 1$
for any $\ell$ and $t$.
In a similar way, we define $n(\ell, t, h)$ the maximum possible number of
$t$-tuples in $L(J_1, \dotsc, J_m)$ each containing a predefined set of $h$
pages. The key part of this section is the proof of the following lemma.

\begin{lemma}\label{lem:n(l,t,h)}
Let $I$ be an interval at level $\ell\geq 2$, and let $J_1, \dotsc, J_m$ be its
children.
The number $n(\ell,t,h)$ of distinct $t$-tuples in the joint request list
$L(J_1, \dotsc, J_m)$, each containing $h$ predefined pages, can be bounded as
follows:
\[ n(\ell,t,h) \leq 2^{(\ell-1)(\ell-1+t)^2 \cdot 2^{\ell-1+t-h}}. \]
\end{lemma}

First, we show that this lemma directly implies Theorem~\ref{thm:dichotomy}.

\begin{proof}[Proof of Theorem~\ref{thm:dichotomy}]
Let us denote $J_1, \dotsc, J_m$ the set of children of $I$.
If their joint request list contains only empty set, then there is a feasible
assignment $\alpha$ which gives no label to $I$. In this case, $I$ can be
feasibly labeled by an arbitrary page, and we have $Q = U$.
Otherwise, the feasible label for $I$ are precisely the pages which are
contained in $L(J_1, \dotsc, J_m)$ as $1$-tuples (singletons), whose number is
bounded by $n(k,1)$.
Therefore, using Lemma~\ref{lem:n(l,t,h)}, we get
\[
n(k,1) = n(k,1,0) \leq 2^{(k-1)(k-1+1)^2\, 2^{(k-1+1)}}
	\leq 2^{2^{k+3\log k}}.
\]
The last inequality holds because
$(k-1)k^2 \leq k^3 \leq 2^{3\log k}$.
\end{proof}

Lemma~\ref{lem:n(l,t,h)} is proved by induction in $\ell$.
However, to establish a relation between $n(\ell,t,h)$ and $n(\ell-1,t,h)$,
we use $f(\ell-1,t,h)$ as an intermediate step.
We need the following observation.

\begin{observation}\label{obs:f_vs_n}
Let $I$ be an interval at level $\ell\geq 2$, and let $J_1, \dotsc, J_m$ denote
all its children. Then we have
\begin{equation}\label{eq:f_vs_n}
f(\ell,t,h) \leq (t-h+1)\,n(\ell,t+1,h).
\end{equation}
\end{observation}
\begin{proof}
Observation~\ref{obs:C_to_I_list} already shows that a $t$-tuple
$A$ belongs to $R(I)$ if and only if there is a $(t+1)$-tuple
$B \in R(J_1, \dotsc, J_m)$ such that $A\subset B$.
If $B$ is not an inclusion-wise minimal member of $R(J_1, \dotsc, J_m)$,
then there is $B'\subsetneq B$ in $R(J_1, \dotsc, J_m)$ and a point
$p$ such that $(B'\setminus \{p\}) \subsetneq A$ also belongs to $R(I)$. This
implies that $A$ does not belong to $L(I)$. Therefore we know that each
$t$-tuple $A\in L(I)$ is a subset of some $B \in L(J_1, \dotsc, J_m)$.

On the other hand, it is easy to see that each $B \in L(J_1, \dotsc, J_m)$
contains precisely $t+1$ distinct $t$-tuples, each created by removing one page
from $B$. If we want all of them to contain
a predefined set $P$ of $h$ pages, then surely $P$ has to be contained in $B$,
and there can be precisely $t+1-h$ such $t$-tuples, each of them
equal to $B\setminus \{p\}$ for some $p\in B\setminus P$.
Therefore we have $f(\ell,t,h)\leq (t+1-h)\,n(\ell,t+1,h)$.
\end{proof}

Therefore, our main task is to bound
$n(\ell,t,h)$ with respect to the values of $f(\ell-1,t',h')$.

\paragraph{Two simple examples and the basic idea.}
Let $I$ be an interval at level $\ell$ and $J_1, \dotsc, J_m$ its children.
Each $t$-tuple in the joint request list $L(J_1, \dotsc, J_m)$
needs to be composed of smaller tuples from $L(J_1), \dotsc, L(J_m)$
(see Observation~\ref{obs:list_prod})
whose numbers are bounded by function $f(\ell-1,t',h')$.
However, to make use of the values $f(\ell-1,t',h')$, we need to consider the
ways in which a $t$-tuple could be created.
To illustrate our basic approach, we consider the following simple
situation.

Let us assume that $L(J_1), \dotsc, L(J_m)$ contain only singletons.
Recall that a pair $\{p,q\}$ can belong to $L(J_1, \dotsc, J_m)$
only if each list $L(J_i)$ contains either $p$ or $q$ as a singleton,
see Observation~\ref{obs:list_prod}.
Therefore, one of them must be contained in
at least half of the lists and we call it a ``popular'' page.
Each list has size at most $f(\ell-1,1,0)$, and therefore
there can be at most $2 f(\ell-1, 1,0)$ popular pages contained in the lists
$L(J_1), \dotsc, L(J_m)$.
A fixed popular page $p$, can be extended to a pair $\{p,q\}$ by at most
$f(\ell-1,1,0)$ choices for $q$, because $q$ has to lie in all the lists
not containing $p$.
This implies that there can be at most $2f(\ell-1,1,0)\cdot f(\ell-1,1,0)$
pairs in $L(J_1, \dotsc, J_m)$.

Here is a bit more complicated example.
We estimate, how many $t$-tuples $A$ can be contained in
$L(J_1, \dotsc, J_m)$, such that the following holds:
$A = A_1 \cup A_2$, where $A_1$ is a $t_1$-tuple from
lists $L(J_1), \dotsc, L(J_{m-1})$ and $A_2$ is a $t_2$-tuple from $L(J_m)$.
We denote $h := |A_1\cap A_2|$.
Then, the number of $t$-tuples in $L(J_1, \dotsc, J_m)$ created
from $L(J_1), \dotsc, L(J_m)$ in this way
cannot be larger than $f(\ell-1, t_1, 0) \cdot f(\ell-1, t_2, h)$,
since the choice of $A_1$ already determines the $h$ pages in $A_2$.

However, the $t$-tuples in $L(J_1, \dotsc, J_m)$
can be created in many complicated ways.
To make our analysis possible, we classify each tuple according to its
{\em specification}, which describes the way it was generated
from $L(J_1), \dotsc, L(J_m)$.
The main idea of our proof is to
bound the number of $t$-tuples which correspond to a given
specification. Then, knowing the number of specifications and having the bounds
for each $L(J_i)$ from the induction, we can get an upper
bound for the overall number of $t$-tuples.

\paragraph{Specifications of $t$-tuples.}
For a fixed $t$-tuple $A\in L(J_1, \dotsc, J_m)$,
we construct its specification $S$ as follows.
First, we sort the pages in $A$ lexicographically, denoting them
$p_1, \dotsc, p_t$.
Let $A_1$ be the subset of $A$ contained in the largest number of lists
$L(J_1), \dotsc, L(J_m)$ as a tuple.
Then, by pigeon-hole principle, $A_1$ lies in at least
$1/2^t$ fraction of the lists, since there are only $2^t$ subsets of $A$
and each list has to contain at least one of them.
We define $T_1$ as the set of indices of the pages in $A_1$, i.e.,
$T_1 = \{i\;|\: p_i\in A_1\}$. Set $T_1$ becomes the first part of the
specification $S$.
Having already defined $A_1, \dotsc, A_j$, we choose $A_{j+1}$ from the lists
which do not contain any subset of $A_1\cup \dotsb \cup A_j$.
We choose $A_{j+1}$ to be the tuple which is contained in the largest number of
them and set $T_{j+1} = \{i\;|\: p_i \in A_{j+1}\}$.

This way we get two important properties.
First, $A_{j+1}$ contains at least one page which is not present in
$A_1\cup\dotsb\cup A_j$.
Second, at least $1/2^t$ fraction of lists which do not contain any subset of
$A_1\cup\dotsb\cup A_j$, contain $A_{j+1}$ as a tuple.
We stop after $n_S$ steps, as soon as $A_1\cup\dotsb\cup A_{n_S} = A$
and each of the
lists contains some subset of $A_1\cup\dotsb\cup A_{n_S}$.
We define the specification $S$ as an ordered tuple $S=(T_1, \dotsc, T_{n_S})$.
Note that $n_S\leq t$, since each $A_{j+1}$ contains a page not yet present in
$A_1\cup \dotsb \cup A_j$.

Let us denote $\cS_t$ the set of all possible specifications of $t$ tuples.
The size of $\cS_t$ can be bounded easily: there are at most $t$ sets contained
in each specification, each of them can be chosen from at most $2^t$
subsets of $\{1, \dotsc, t\}$,
implying that $|\cS_t| \leq (2^t)^t = 2^{t^2}$.
Let us denote $n(\ell, S, h)$ the number of $t$-tuples in $L(J_1, \dotsc, J_m)$
having $h$ pages predefined which correspond to the specification $S$.
Since each $t$-tuple $A$ has a (unique) specification, we have the following
important relation:
\begin{equation}\label{eq:n(l,t,h)}
n(\ell, t, h) \leq \sum_{S\in \cS_t} n(\ell, S, h).
\end{equation}

\paragraph{Number of $t$-tuples per specification.}
First, let us consider a simpler case when $h = 0$.
Let $S = (T_1, \dotsc, T_{n_S})$ be a fixed specification of $t$-tuples.
For each $j=1,\dotsc, n_S$, we define
$t_j := |T_j|$, and $d_j := |T_j \setminus \bigcup_{i=1}^{j-1} T_i|$ the number
of new indices of $T_j$ not yet contained in the previous sets $T_i$.
There can be at most $2^t f(\ell-1, t_1, 0)$ choices for a
$t_1$-tuple $A_1$ corresponding to the indices of $T_1$.
This can be shown by a volume argument: each such tuple has to be
contained in at least $1/2^t$ fraction of $L(J_1),\dotsc, L(J_m)$, and each
list can contain at most $f(\ell-1, t_1, 0)$ $t_1$-tuples.
By choosing $A_1$, some of the request lists are already covered, i.e., the
ones that contain $A_1$ or its subset as a tuple.
According to the specification, $A_2$ has to be contained in at least $1/2^t$
fraction of the lists which are not yet covered by $A_1$.
However, the choice of $A_1$ might have already determined some pages of
$A_2$ unless $t_2=d_2$.
Therefore, the number of choices for $A_2$ can be at
most $2^t f(\ell-1, t_2, t_2-d_2)$.
In total, we get the following bound:
\begin{equation}
n(\ell, S, 0) \leq \prod_{i=1}^{n_S} 2^t f(\ell-1, t_i, t_i-d_i).
\end{equation}

For the inductive step, we also need to consider the case when some
pages of the tuple are fixed in advance.
Let $P$ be a predefined set of $h$ pages.
We want to bound the maximum possible number of $t$-tuples
containing $P$ in $L(J_1, \dotsc, J_m)$.
However, two different $t$-tuples containing $P$ can have the pages of $P$
placed at different indices, which affects the number of pre-fixed indices in
each $T_i$.
Therefore, we first choose the set $C$ of $h$ indices which
will be occupied by the pages of $P$. There are $\binom{t}{h}$ choices for $C$,
and, by definition of the specification,
the pages of $P$ have to be placed at those indices in alphabetical order.
For a fixed $C$, we denote $\bar{d}_i$ the number of not predetermined indices
contained in $T_i$, i.e.
$\bar{d}_i := |T_i \setminus (C \cup T_1 \cup \dotsb \cup T_{j-1})|$.
We get the following inequality:
\begin{equation}\label{eq:n(l,S,h)}
n(\ell, S, h)
	\leq \sum_{C\in \binom{[t]}{h}}
		\prod_{i=1}^{n_S} 2^t f(\ell-1, t_i, t_i - \bar{d}_i)
	\leq \sum_{C\in \binom{[t]}{h}} 2^{t^2}
		\prod_{i=1}^{n_S} f(\ell-1, t_i, t_i - \bar{d}_i).
\end{equation}
Now we are ready to prove Lemma~\ref{lem:n(l,t,h)}.

\paragraph{Proof of Lemma~\ref{lem:n(l,t,h)}.}
Combining equations~\eqref{eq:n(l,t,h)} and~\eqref{eq:n(l,S,h)},
we can bound $n(\ell, t, h)$ with
respect to $f(\ell-1, t', h')$. We get
\begin{equation}\label{eq:n(l,t,h)_final}
n(\ell, t, h) \leq \sum_{S\in\cS_t} \sum_{C\in \binom{[t]}{h}}
	2^{t^2} \prod_{i=1}^{n_S} f(\ell-1, t_i, t_i-\bar{d}_i).
\end{equation}
Now, we proceed by induction.
We bound this quantity using Observation~\ref{obs:f_vs_n} with respect to values
of $n(\ell-1,t',h')$, which we know from the inductive hypothesis.

In the base case $\ell=2$, we use \eqref{eq:n(l,t,h)_final}
to bound the value of $n(2,t,h)$.
Here, we have $f(1,t',h') = t'-h'+1$ because of the following reason.
If a leaf interval $I$ has a $t'$-tuple in its request list
$L(I)$, there must be a set $Q$ of $t'+1$ distinct pages requested during the
time interval of $I$. Then, $L(I)$ contains precisely $t'+1$ distinct
$t'$-tuple depending on which page becomes a label of $I$. Those $t$-tuples are
$Q\setminus \{q\}$ for $q\in Q$.
However, if we count only $t'$-tuples which contain some predefined
set $P\subseteq Q$ of $h'$ pages, there can be only $t'+1-h'$ of them,
since $Q\setminus \{q\}$ contains
$P$ if and only if $q$ does not belong to $P$.
Thereby, for any choice of $S$ and $C$, we have
$f(1,t_i, t_i-\bar{d}_i) = t_i - (t_i-\bar{d}_i) + 1 \leq t-h+1$,
since $t-h = \sum_{i=1}^{n_S}\bar{d}_i$. If $t=h$, we clearly have
$n(\ell,t,h) = 1$.
Otherwise, use \eqref{eq:n(l,t,h)_final} with the following estimations applied:
the size of $\cS_t$ is at most $2^{t^2}$,
the number of choices for $C$ is at most $t^h$,
and $n_S$ is at most $t$. We get
\[ n(2, t, h)\leq 2^{t^2}\, t^h\, 2^{t^2}\, (t-h+1)^t
	\leq 2^{4t^2} \leq 2^{t^2\, 2^{1+t-h}},
\]
where the first inequality holds since both $t^h$ and $(t-h+1)^t$ can be bounded
by $2^{t^2}$. The last inequality follows, since $2^{1+t-h}\geq 4$,
and this concludes the proof of the base case.

Now we proceed to the case $\ell > 2$.
For a fixed $S$ and $C$, we bound the product inside
equation~\eqref{eq:n(l,t,h)_final}, and our goal is to get a bound independent
on the particular choice of $S$ and $C$. Using Observation~\ref{obs:f_vs_n},
we get
\[ \prod_{i=1}^{n_S} f(\ell-1, t_i, t_i-\bar{d}_i)
	\leq \prod_{i=1}^{n_S} (t_i+1)\,n(\ell-1, t_i+1, t_i-\bar{d}_i).
\]
Now, we take the logarithm of this inequality and apply the
inductive hypothesis. We get
\[ \log \prod_{i=1}^{n_S} f(\ell-1, t_i, t_i-\bar{d}_i)
	\leq \sum_{i=1}^{n_S} \log(t_i+1)
		+ \sum_{i=1}^{n_S} (\ell-2)(\ell-2+t_i+1)^2\,
		2^{(\ell-2) + (t_i+1) - (t_i-\bar{d}_i)}.
\]
This is at most
$t\log(t+1) + (\ell-2)(\ell-1+t)^2\,2^{\ell-1} \sum_{i=1}^{n_S} 2^{\bar{d}_i}$,
where the last sum cannot be larger than $2^{t-h}$, since we have
$\sum_{i=1}^{n_S} \bar{d}_i = t-h$ and $\sum 2^{x_i} \leq 2^{\sum x_i}$.
Now, we can get rid of all $t_i$ and $\bar{d}_i$ which are dependent on the
choice of $S$ and $C$. We can write the preceding inequality as follows:
\begin{equation}\label{eq:log_prod_f}
\log \prod_{i=1}^{n_S} f(\ell-1, t_i, t_i-\bar{d}_i)
	\leq t\log(t+1) + (\ell-2)(\ell-1+t)^2\,2^{\ell-1 + t - h}.
\end{equation}
To finish the proof, we plug the bound from \eqref{eq:log_prod_f} to
\eqref{eq:n(l,t,h)_final}:
\[ n(\ell, t, h) \leq |\cS_t|\cdot t^h \cdot 2^{t^2} \cdot
	2^{t\log(t+1) + (\ell-2)(\ell-1+t)^2\,2^{\ell-1 + t - h}},
\]
where the size of $\cS_t$ is at most $2^{t^2}$. Taking the
logarithm of this inequality, we get
\[ \log n(\ell,t,h) \leq t^2 + h \log t + t^2 + t\log(t+1)
		+ (\ell-2)(\ell-1+t)^2\,2^{\ell-1 + t - h}
	\leq (\ell-1)(\ell-1+t)^2\,2^{\ell-1 + t - h},
\]
what already implies the statement of the lemma.
To see why the last inequality holds,
note that $(\ell-1+t)^2$ is at least $t^2$, and
$2^{\ell-1 + t - h}$ is always greater than $2^2$, since we have
$\ell \geq 3$.
Therefore, the sum of the four smaller-order terms can be bounded by
$(\ell-1+t)^2\,2^{\ell-1 + t - h}$.
\hfill\qedsymbol

\subsection{Case of $d$ different weights}
\label{sec:d-dichotomy}
Now we prove a dichotomy theorem for the case of $d$ different weight classes.
For each $i= 1, \dotsc, d$, let  $k_i$ denote the number of server of weight
$w_i$, so that $k = k_1 + \dotsb + k_d$.
In the rest of this section we assume that $k_1, \dotsc, k_d$ are fixed
and our estimations of $f(\ell, t, h)$ and $n(\ell, t, h)$ will implicitly
depend on their values.
We consider a service pattern $\cI$ consisting of $d$
levels $\cI_1, \dotsc, \cI_d$, where an interval at level $i$ has a label
consisting of at most $k_i$ different pages describing the position of the $k_i$
servers of weight $w_i$.
The cost of $\cI$ is computed as $\sum_{i=1}^d k_iw_i (|\cI_i|-1)$, and the
assignment function $\alpha$ labels each interval in $\cI_i$ with a set
$C_i$ of at most $k_i$ points.

The definition of the request sets and the request lists stays similar to the
general setting.
We say that a tuple of pages $S$ belongs to the request set $R(I)$ of an
interval $I$, if there is a feasible assignment $\alpha$ which labels the
ancestors of $I$ only using pages of $S$, and again we define $L(I)$ to be the
set of inclusion-wise minimal tuples from $R(I)$.
Observation~\ref{obs:list_prod} holds as it is stated in
Section~\ref{sec:intervals}. For the Observation~\ref{obs:C_to_I_list}, we have
the following variant.
\begin{observation}\label{obs:Cd_to_I_list}
Let $J_1, \dotsc, J_m$ denote all the children of some $\ell$th level
interval $I$.
A tuple $Q$ belongs to $R(I)$ if and only if there is a set $C$ of $k_\ell$
pages such that  $Q \cup C$ belongs to $R(J_1, \dotsc, J_m)$.
\end{observation}

The statement of the theorem for this case is a bit more complicated due to the
following phenomenon.
Let $I$ be a top-level interval such that the joint request list of its children
contains precisely one singleton $\{p\}$.
Then any $k_d$-tuple can be feasibly assigned to
$I$, whenever it contains $p$. This way there is potentially infinite number of
feasible labels for $I$, but the labels are not yet
arbitrary and they all have to contain $p$
what makes them easy to identify.
Therefore we state the theorem in the following way.

\begin{theorem}[Dichotomy theorem for $d$ different weights]
\label{thm:dichotomy-d}
Let $\cI$ be a service pattern, $I\in \cI$ be an arbitrary interval at
level $d$ and let us denote $\cQ = \cQ_1 \cup \dotsb \cup \cQ_{k_d}$ a set
of labels for $I$ satisfying the following:
\begin{itemize}
\item Each $\cQ_t$ contains feasible labels $T$ for $I$, such that $|T|=t$.
\item Whenever $T$ is in $\cQ_t$, no $T'$ containing $T$ as a subset
	belongs to any $\cQ_j$ for $j>t$.
\end{itemize}
Then, either $Q_1 = U$, or $|Q_t| \leq n(d,t)$ for each $t$, where
$n(d,t) \leq 2^{4dk^2 t \prod_{j=1}^{d-1}(k_j+1)}$.
\end{theorem}

If $\cQ_1 = U$, any label $C \in \binom{U}{k_d}$ can be feasibly
assigned to $I$. The crucial part of the proof is the following lemma
that bounds the size of the request list in each level.
It is proved similarly as Lemma~\ref{lem:n(l,t,h)},
although the recursion and the resulting
bounds have a different form.

\begin{lemma}\label{lem:n(ld,t,h)}
Let $I$ be an interval at level $\ell\geq 2$ and $J_1, \dotsc, J_m$
be its children.
The number $n(\ell, t, h)$ of distinct
$t$-tuples in their joint list $L(J_1, \dotsc, J_m)$ having $h$ pages
fixed satisfies:
\[ n(\ell, t, h)
	\leq 2^{\ell\cdot 4k^2 (t-h)\prod_{i=1}^{\ell-1}(k_i + 1)}. \]
\end{lemma}

First, let us show that this lemma already implies
Theorem~\ref{thm:dichotomy-d}.

\begin{proof}[Proof of Theorem~\ref{thm:dichotomy-d}]
Let us denote $J_1, \dotsc, J_m$ the children of $I$.
If their joint request list contains only empty set, then there is a feasible
assignment $\alpha$ which gives no label to $I$. In this case, $I$ can be
feasibly labeled by an arbitrary singleton and we have $\cQ_1 = U$.
Otherwise, a feasible label for $I$ can only be some tuple which is contained
in $L(J_1, \dotsc, J_m)$, and there can be at most $n(k,t)$ $t$-tuples in
$L(J_1, \dotsc, J_m)$.
Lemma~\ref{lem:n(ld,t,h)} implies, that the number $n(k,t)$ fulfills the bound
stated by the theorem:
$ n(k, t) = n(k,t,0) \leq 2^{d\cdot 4k^2 t \prod_{i=1}^{d-1}(k_i + 1)}. $
\end{proof}

To prove Lemma~\ref{lem:n(ld,t,h)}, we proceed by induction in $\ell$.
First, we establish the relation between $f(\ell,t,h)$ and $n(\ell,t,h)$.
\begin{observation}\label{obs:fd_vs_n}
Let $I$ be an interval at level $\ell\geq 2$, and $J_1, \dotsc, J_m$ its
children. Then we have
\begin{equation}
f(\ell,t,h) \leq \binom{t+k_\ell-h}{k_\ell}\,n(\ell,t+1,h).
\end{equation}
\end{observation}
\begin{proof}
Observation~\ref{obs:Cd_to_I_list} already shows that a $t$-tuple
$A$ belongs to $R(I)$ if and only if there is a $(t+k_\ell)$-tuple
$B \in R(J_1, \dotsc, J_m)$ such that $A\subset B$.
If $B$ is not an inclusion-wise minimal member of $R(J_1, \dotsc, J_m)$,
then there is some $B'\subsetneq B$ in $R(J_1, \dotsc, J_m)$ and a set $C$ of
$k_\ell$ pages such that $(B'\setminus C) \subsetneq A$ also belongs to $R(I)$.
This implies that $A$ does not belong to $L(I)$. Therefore we know that each
$t$-tuple $A\in L(I)$ is a subset of some $B \in L(J_1, \dotsc, J_m)$.

On the other hand, each $B \in L(J_1, \dotsc, J_m)$
contains precisely $\binom{t+k_\ell}{k_\ell}$ distinct $t$-tuples,
each created by removing $k_\ell$ page from $B$.
If we want all of them to contain a predefined set $P$ of $h$ pages,
then $P$ has to be contained in $B$,
and there can be precisely $\binom{t+k_\ell-h}{k_\ell}$ such $t$-tuples,
each of them equal to $B\setminus C$ for some
$C\subset B\setminus P$ of size $k_\ell$.
Therefore we have $f(\ell,t,h)\leq \binom{t+k_\ell-h}{k_\ell}\,n(\ell,t+1,h)$.
\end{proof}

To bound $n(\ell,t,h)$ with respect to numbers $f(\ell-1,t,h)$,
we use specifications as they are defined the previous subsection,
so that we have
\begin{equation}\label{eq:n(ld,t,h)}
n(\ell,t,h) \leq \sum_{S\in \cS_t} n(\ell, S, h),
\end{equation}
and also
\begin{equation}\label{eq:n(ld,S,h)}
n(\ell,S,h) \leq \sum_{C\in \binom{[t]}{i}}
	2^{t^2} \prod_{i=1}^{n_S} f(\ell-1, t_i, t_i - \bar{d}_i).
\end{equation}
Therefore we can proceed directly to the proof.

\begin{proof}[Proof of Lemma~\ref{lem:n(ld,t,h)}]
Combining equations \eqref{eq:n(ld,t,h)} and \eqref{eq:n(ld,S,h)}, we can bound
$n(\ell, t, h)$ with respect to the values of $f(\ell-1,t',h')$. We get
\begin{equation}\label{eq:n(ld,t,h)-final}
n(\ell,t,h) \leq \sum_{S\in \cS_t} \sum_{C\in \binom{[t]}{i}}
	2^{t^2} \prod_{i=1}^{n_S} f(\ell-1, t_i, t_i - \bar{d}_i).
\end{equation}
In the rest of the proof, we use induction to show
that this inequality together with Observation~\ref{obs:Cd_to_I_list}
implies the desired bound.

In the base case, we have $\ell=2$, and we can use \eqref{eq:n(ld,t,h)-final}
directly with the values of $f(1,t',h')$.
We know that $f(1,t',h') = \binom{t'-h'+k_1}{k_1}$,
for the following reason.
To obtain a tuple of length $t'$ in the request list of a first level interval
$I$, there must be $t'+k_1$ distinct points requested in the input
sequence during this interval.
As $h'$ of them are pre-specified to be contained in each tuple, we are left
with $t'-h'+k_1$ points from which the label for $I$ is chosen, and this label
contains precisely $k_1$ points. Therefore, $L(I)$ can contain at most
$\binom{t'-h'+k_1}{k_1}$ distinct $t'$-tuples which contain the $h'$
predetermined points.
Therefore, for each $i$ in the product in \eqref{eq:n(ld,t,h)-final}, we have
$f(1, t_i, t_i - \bar{d}_i) \leq \binom{t_i - (t_i-\bar{d}_i) + k_1}{k_1}
\leq \binom{\bar{d}_i + k_1}{k_1} \leq (t-h+k_1)^{k_1}$.
However, $\sum_{i=1}^{n_S} \bar{d}_i = t-h$ and $f(1,t_i,t_i-\bar{d}_i)=1$
whenever $\bar{d}_i=0$. Therefore at most $t-h$ factors in that product can be
greater than $1$, and we have
$\prod_{i=1}^{n_S} f(\ell-1, t_i, t_i - \bar{d}_i) \leq (t-h+k_1)^{k_1(t-h)}$.
Recall that there are at most $|\cS_t|\leq 2^{t^2}$ choices for $S$ and at most
$t^h$ choices for $C$. Using the trivial estimate $h\leq t\leq k$, we get
\[ n(2,t,h) \leq 2^{t^2}\, t^h\, 2^{t^2}\, (t-h+k_1)^{k_1(t-h)}
	\leq 2^{t^2 +h\log t + t^2 + k_1(t-h) \log(t-h+k_1)}
	\leq 2^{2\cdot 4k^2 (t-h) (k_1+1)},
\]
where the right-hand side corresponds to the bound claimed by the lemma.

Let us now focus on the inductive step with $\ell>2$.
For a fixed $S$ and $C$, we bound the product inside
equation~\eqref{eq:n(ld,t,h)-final} by an expression independent on $S$ and $C$.
First, let us apply Observation~\ref{obs:fd_vs_n} to each term of the product.
Since
$\binom{t_i+k_{\ell-1} - (t_i-\bar{d}_i)}{k_{\ell-1}}
	= \binom{k_{\ell-1} + \bar{d}_i}{k_{\ell-1}}$,
we have
\[ f(\ell-1,t_i,t_i-\bar{d}_i)
	\leq \binom{k_{\ell-1}+\bar{d}_i}{k_{\ell-1}}
		n(\ell-1, t_i+k_{\ell-1}, t_i-\bar{d}_i).
\]
Let us now consider the logarithm of this inequality. Bounding
$\binom{k_{\ell-1}+\bar{d}_i}{k_{\ell-1}}$ by
$(k_{\ell-1}+t-h)^{k_{\ell-1}}$ and
applying the inductive hypothesis, we get
\[ \log f(\ell-1,t_i,t_i-\bar{d}_i) \leq k_{\ell-1}\log(k_{\ell-1}+t-h)
	+ (\ell-1)\cdot 4k^2\big(t_i+k_{\ell-1} -(t_i-\bar{d}_i)\big)
		\prod_{j=1}^{\ell-2} (k_j+1).
\]
Note that the only term in this bound, which is dependent on $i$, is
$(t_i+k_{\ell-1} -(t_i-\bar{d}_i)) = (\bar{d}_i+k_{\ell-1})$.
Now we would like to bound
$\log \prod_{i=1}^{n_S} f(\ell-1, t_i, t_i - \bar{d}_i)$, which equals to a sum
of $\log f(\ell-1,t_i,t_i-\bar{d}_i)$ over $i=1, \dotsc, n_S$.
First, note that $f(\ell,t_i,t_i-\bar{d}_i)=1$ whenever $\bar{d}_i=0$.
Let us denote $A$ the set of indices $i$ such that $\bar{d}_i>0$.
Then, by the inequality above,
\[ \log \prod_{i\in A} f(\ell-1, t_i, t_i - \bar{d}_i)
	\leq |A|\cdot k_{\ell-1}\log(t+k_{\ell-1}-h)
		+ (\ell-1)\cdot 4k^2\prod_{j=1}^{\ell-2} (k_j+1)
		\cdot \sum_{i\in A}(\bar{d}_i + k_{\ell-1}).
\]
We know that $\sum_{i=1}^{n_S} \bar{d}_i = t-h$, what implies that the size of
$A$ is also at most $t-h$. Therefore we can bound the last sum as follows:
$\sum_{i\in A}(\bar{d}_i + k_{\ell-1}) \leq (t-h)+|A|\cdot k_{\ell-1}
	\leq (t-h)(k_{\ell-1}+1)$.
Since $f(\ell,t_i,t_i-\bar{d}_i)=1$ for each $i\notin A$, we get the following
bound:
\[ \prod_{i=1}^{n_S} f(\ell-1, t_i, t_i - \bar{d}_i)
	= \prod_{i\in A} f(\ell-1, t_i, t_i - \bar{d}_i)
	\leq (t+k_{\ell-1}-h)^{(t-h)k_{\ell-1}}\,
	2^{(\ell-1) \cdot 4k^2(t-h)\prod_{j=1}^{\ell-1} (k_j+1)}.
\]
Now we are almost done.
The preceding bound is universal and independent of $S$ and $C$,
and therefore we can plug it in the equation \eqref{eq:n(ld,t,h)-final} in the
following way:
\[ n(\ell,t,h)\leq |\cS_t|\cdot t^h \cdot 2^{t^2}\cdot
	(t+k_{\ell-1}-h)^{(t-h)k_{\ell-1}}\cdot
	2^{(\ell-1)\cdot 4k^2(t-h)
		\prod_{i=1}^{\ell-1}(k_i+1)}, \]
where the size of $\cS_t$ is at most $2^{t^2}$.
It is enough to show that the last term is much larger and dominates the
smaller terms. To show this, we take the logarithm of this
inequality and we get
\[ \log n(\ell,t,h) \leq t^2 + h\log t + t^2 +
	(t-h) k_{\ell-1}\log(t+k_{\ell-1}-h)
		+ (\ell-1)\cdot 4k^2 (t-h) \prod_{i=1}^{\ell-1}(k_i+1), \]
where each of the four smaller-order terms is smaller than
$k^2(t-h)\prod_{i=1}^{\ell-1}(k_i+1)$.
Therefore, we get the final inequality which concludes the proof:
$n(\ell,t,h)\leq 2^{\ell\cdot 4k^2(t-h) \prod_{i=1}^{\ell-1}(k_i+1)}.$
\end{proof}

\section{Upper bounds for generalized WFA}
\label{sec:ubs}
We now show that the generalized Work
Function Algorithm with $\lambda=0.5$ achieves the bounds claimed in Theorems~\ref{thm:ub} and \ref{thm:ubd}.
Even though Theorem \ref{thm:ub} follows as a special case of Theorem \ref{thm:ubd} (up to lower order terms in the exponent) we describe these results
separately in Sections \ref{sec:ub} and \ref{sec:ubd} as the proof of Theorem \ref{thm:ub} is simpler and highlights the main ideas directly.

\subsection{Upper bound for arbitrary weights}
\label{sec:ub}
We prove Theorem~\ref{thm:ub} by induction on the number of servers. Let $r_k$ denote the bound on the competitive ratio with $k$ servers. We will show that $r_k = O((n_k)^3 r_{k-1})$, where $n_k$ is the constant from the Dichotomy Theorem~\ref{thm:dichotomy}. As $r_1=1$ trivially, this will imply the result.
We begin with some definitions and the basic properties of WFA.

\paragraph{Definitions and Notation.}
Recall the definition of Work functions and the generalized WFA from Section \ref{sec:prelim}. A basic property of work functions is that for any two configurations $C$ and $C'$ and any time $t$, the work function values $\WF_t(C)$ and $\WF_t(C')$ can differ by at most $d(C,C')$. Moreover, at any time $t$, the generalized WFA will always be in some configuration that contains the current request $\sigma_t$. For the rest of this section we focus on $\WFA_{0.5}$ and denote it by $\ALG$.

Let $M_t$ denote the minimum work function value at
time $t$ over all configurations, and
let $\WF_t(p) = \min\{\WF_t(C)\;|\, C(k) = p\}$ denote the minimum work function value over all configurations with the heaviest server $s_k$ at $p$.
We denote $W_i = \sum_{j=1}^i w_i$.
We will assume (by rounding if necessary) that the weights $w_i$ are well-separated and satisfy $W_{i-1} \leq w_i/(20 i n_i)$ for each $i=2,\ldots,k$.
This can increase the competitive ratio by at most $O(k^k \Pi_{i=1}^k n_i) \ll
O(n_k^3)$.  This will ensure that for any two configurations $C$ and $C'$ that
both have $s_k$ at $p$, their work function values differ by at most $W_{k-1}$
which is negligibly small compared to $w_k$.

For a point $p \in U$, we define the {\em ``static'' work function} $\SWF_t(p)$ as the optimal cost to serve requests $\sigma_1, \dotsc, \sigma_t$ while keeping server $s_k$ fixed at point $p$. Note that this function will in general take very different values than the (usual) work function.
However, the local changes of $\SWF(p)$ will be useful in our inductive argument. Intuitively, if $\ALG$ keeps $s_k$ at $p$ during some interval $[t_1,t_2]$ and $\SWF(p)$ rises by $x$ during this period, then the cost incurred by $\ALG$ should be at most $r_{k-1} x$.

For any quantity $X$, we use $\Delta_{t_1}^{t_2} X:=X_{t_2} - X_{t_1}$ to denote the change in $X$ during the  time interval $[t_1,t_2]$. If the time interval is clear from the context, we use $\Delta X$.

We partition the request sequence into {\em phases}, where a phase ends whenever $\ALG$ moves its heaviest server $s_k^{\ALG}$.

\paragraph{Basic Properties of WFA.}
We describe some simple facts that follow from basic properties of $\WFA_{\lambda}$ and work functions. The proofs of the following lemmas are in Appendix~\ref{app_sec:ub}.

\begin{lemma}
\label{cl:phase_start_end}
Consider a phase that starts at time $t_1$ and end at $t_2$, and let $p$ be the location of $s_k^{\ALG}$ during this phase. Then,
\begin{enumerate}[(i)]
\item $M_{t_1}\leq \WF_{t_1}(p) \leq M_{t_1} + W_{k-1} $, and
\item $ w_k/2 - 2W_{k-1} \leq \Delta \WF(p) \leq \Delta M + w_k/2 + 2 W_{k-1}.$
\end{enumerate}
\end{lemma}

The next lemma shows that $\WF(p)$ and $\SWF(p)$ increase by similar amount while $s_k^{\ALG}$ remains at point $p$.

\begin{lemma}
\label{lem:swf_wf_phase}
For a phase where $s_k^{\ALG}$ is at point $p$, we have that $
|\Delta \WF(p) -  \Delta \SWF(p) | \leq  W_{k-1}$.
\end{lemma}
We remark that the preceding lemma does not hold for some $q$ where $s_k^{\ALG}$ is not present.
The following lemma is more general and holds for any point $p \in U$ and for any time interval, even if there are many phases in between.
\begin{lemma}\label{lem:WFvsSWF} For any $t'>t$, $p \in U$,
\[ \WF_{t'}(p) \geq \min\{\WF_t(p) + \Delta_{t}^{t'}\SWF(p)
	- W_{k-1}, M_t + w_k\}.
\]
\end{lemma}

\paragraph{Bounding the Performance.}
We are now ready to prove Theorem \ref{thm:ub}. The key lemma will be the following.

\begin{lemma}
\label{lem:main_gen}
Consider any sequence of $m = n_k +1$ consecutive phases. Then, $\Delta M \geq w_k/(8 k \cdot n_k) $ and the cost incurred by $\ALG$ is at most $4 n_k \cdot r_{k-1} \cdot w_k + r_{k-1} \cdot \Delta M $.
\end{lemma}

Before proving Lemma \ref{lem:main_gen}, let us see why gives a competitive ratio $r_k = O(n_k^3) \cdot r_{k-1}$, and hence proves Theorem \ref{thm:ub}.

\paragraph*{Proof of Theorem~\ref{thm:ub}}
Let $\cost(\ALG)$ and $\cost(\OPT)$ denote the cost of the algorithm and the optimal cost respectively. We show that $\ALG$ with $k$ servers is strictly $r_k$-competitive, i.e. $\cost(\ALG) \leq r_k \cdot \cost(\OPT)$ for any request sequence, given that $\ALG$ and $\OPT$ start from the same initial configuration.

For $k=1$, $\ALG$ is obviously strictly 1-competitive. Assume inductively that $\ALG$ with $k-1$ servers is strictly $r_{k-1}$-competitive. We now bound $r_k$.

Let $m$ denote the total number of phases. We partition the sequence into $h = \lceil \frac{m}{n_k+1} \rceil $ groups where each group (except possibly the last one) consists of $n_{k}+1$ phases. Note that $\cost(\OPT) = M_T$, where $M_T$ is the minimum work function value at
the end of the request sequence. Thus for each group of phases we can use $\Delta M$ as an estimate of the optimal cost. 

{\em Competitive Ratio:} We first show that $\ALG$ is $r_k$-competitive, since this proof is simple and highlights the main idea. We then give a more careful analysis to show that in fact $\ALG$ is strictly $r_k$-competitive.  
 
By Lemma~\ref{lem:main_gen}, during $i$th group, $ i \leq h-1$, the ratio between the cost of $\ALG$ and $\Delta M$ is at most
\begin{equation}
\label{eq:ratio_group}
\frac{4 n_k \cdot w_k \cdot r_{k-1} + r_{k-1} \Delta M}{\Delta M} \leq
 \frac{4 n_k \cdot w_k \cdot r_{k-1}}{w_k / (8 k \cdot n_k)} + \frac{r_{k-1} \Delta M}{\Delta M} \leq 33 k \cdot n_k^2 \cdot r_{k-1}.
\end{equation}
Due to Lemma~\ref{lem:main_gen}, we have that for the last group of phases the cost of $\ALG$ is at most $ 4 n_k \cdot r_{k-1} \cdot w_k + r_{k-1} \cdot \Delta M $.
Overall, we get that $\cost(\ALG) \leq r_k \cdot  M_T + 4 n_k \cdot r_{k-1} \cdot w_k$, for some $r_k = O((n_k)^3 r_{k-1}) $, i.e. $\ALG$ is $r_k$-competitive.

{\em Strict Competitive Ratio:} In order to prove strict competitive ratio, we need to remove the additive term due to the last group of phases. In case $h \geq 2$, we do that by considering the last two groups together. By a similar calculation as in~\eqref{eq:ratio_group} we get that during groups $h-1$ and $h$, the ratio between cost of $\ALG$ and $\Delta M$ is at most $65k n_k^2 \cdot r_{k-1}$. For $i$th group, $i \leq h-2$ we use inequality~\eqref{eq:ratio_group}. Thus, in case $h \geq 2 $ we get that

\begin{equation}
 \frac{\cost(\ALG)}{M_T} \leq 65k n_k^2 \cdot r_{k-1} = O(n_k^3 r_{k-1}).
\end{equation} 

It remains to consider the case $h=1$, i.e there are no more than $n_k+1$ phases. To this end, we distinguish between two cases.

\begin{enumerate}
\item $\OPT$ moves $s_k$: Then $\cost(\OPT) = M_T \geq w_k$ and by Lemma~\ref{lem:main_gen}, $\cost(\ALG) \leq 4 n_k \cdot r_{k-1} \cdot w_k + r_{k-1} \cdot M_T  $. We get that 
\begin{align*}
 \frac{\cost(\ALG)}{\cost(\OPT)} &\leq  \frac{4 n_k \cdot r_{k-1} \cdot
 w_k}{M_T} + \frac{r_{k-1} \cdot M_T}{M_T} \leq \frac{4 n_k \cdot r_{k-1} \cdot
 w_k}{w_k} + r_{k-1} \ll 65k n_k^2 \cdot r_{k-1}
\end{align*}

\item $\OPT$ does not move $s_k$: In this case, $M_T = \WF_T(p_1)$, where $p_1$ is the initial location of the heaviest server $s_k$. We consider 2 sub-cases. 

\begin{enumerate}
\item First phase never ends: In this case, both $\ALG$ and $\OPT$ use $k-1$ servers and start from the same initial configuration, so by the inductive hypothesis $\cost(\ALG) \leq r_{k-1} \cdot \cost(\OPT)$. 
\item First phase ends: By Lemma~\ref{cl:phase_start_end}, we have that for the
first phase $ \Delta \WF(p_1)  \geq w_k/2 - 2W_{k-1} \geq w_k/4 $. Thus we get
that $\WF_T(p_1) \geq w_k/4$, which by a calculation similar
to~\eqref{eq:ratio_group} gives that $ \cost(\ALG) / \cost(\OPT) \leq  17 n_k
r_{k-1} \ll 65k n_k^2 \cdot r_{k-1}$.
\end{enumerate}
\end{enumerate}
We conclude that for any request sequence
\begin{equation}
r_k \leq \frac{\cost(\ALG)}{M_T} \leq 65k n_k^2 \cdot r_{k-1}.
\end{equation}

{\em Calculating the Recurrence.}
Assuming that $r_{k-1} \leq 2^{2^{k+5\log k }}$, and as $n_k = 2^{2^{k+3 \log k}}$ and $\log 65k < 2^{k + 3\log k}$, it follows that
\[
\pushQED{\qed}
\log r_k \leq \log (65k) + 2^{k+ 3 \log k +1} + 2^{k+5\log k } \leq  2^{k +1 + 5
\log (k+1)}.\qedhere
\popQED
\]

We now focus on proving Lemma \ref{lem:main_gen}.
The crucial part is to lower bound the increase in $\Delta M$ during the $m$ phases. Let $t_1$ and $t_2$ denote the start and end times of the $m$ phase sequence. We will show that for all points $p$, $\WF_{t_2}(p) \geq M_{t_1} + w_k/(8 k \cdot n_k)$.
To do this, we upper bound the number of points $p$ where the increase in $\WF(p)$ could be very small in the first phase (Lemma~\ref{lem:egalitarians_cost}). Then, using Lemma~\ref{cl:phase_start_end} we show that, during the subsequent $m$ phases, $s_k^{\ALG}$ will visits all such points $p$ which would increase $\WF(p)$ significantly for each of them.
We now give the details.

Call a point $q$ {\em lucky} during a phase, if its static work function increases by at most $\Delta \SWF(q) < w_k/(4kn_k)$ during that phase. The next lemma shows that there cannot be too many lucky points during a phase.

\begin{lemma}\label{lem:egalitarians_cost}
Let $L$ be the set of lucky points during some phase. Then, $|L| \leq n_k$.
\end{lemma}
\begin{proof}
For the sake of contradiction, suppose that $|L| > n_k$.
Let $Q$  be an arbitrary subset of $L$ such that $|Q| = n_k+1 $. For each $q\in Q$, let $\cI^q$ be the optimal service pattern for the
phase where $s_k$ remained at $q$ throughout. Clearly, $\cost (\cI^q) \leq \Delta \SWF(q)$.

We create a new service pattern $\cI$ that is a {\em refinement} of all $\cI^q$, for $q \in Q$ as follows.
For each $\ell=1, \dotsc, k$, we set
$\cI_\ell = \{ [t_i, t_{i+1})\;|\, \text{ for } i = 1, \dotsc, s-1\}$, where
$t_1 < \dotsb < t_s$ are the times when at least one interval from
$\cI^1_\ell, \dotsc, \cI^{|Q|}_\ell$ ends. This way, each interval $I\in \cI^q_\ell$
is a union of some intervals from $\cI_\ell$.
Let $\cI = \cI_1 \cup \dotsb \cup \cI_k$. Note that any feasible labeling $\alpha$ for any $\cI^q$, extends naturally to a feasible labeling for $\cI$: If an interval $I \in \cI^q $ is partitioned into smaller intervals, we label all of them with $\alpha(I)$.

We modify $\cI$ to be hierarchical, which increases its cost at most by a factor of $k$.
By construction, we have
\begin{equation}\label{eq:cost_refinement}
    \cost(\cI) \leq k \cdot \sum_{q \in Q} \cost (\cI^q)
               \leq k \cdot \sum_{q \in Q} \Delta \SWF(q)
               \leq k(n_k+1) \cdot \frac{w_k}{4kn_k} \leq \frac{w_k}{3}.
\end{equation}

Now the key point is that $\cI$ has only one interval $I$ at level $k$, and all $q \in Q$ can be feasibly assigned to it.
But by the Dichotomy theorem~\ref{thm:dichotomy}, either the number of points which
can be feasibly assigned to $I$ is at most $n(k,1)$, or else
any point can be feasibly assigned there. As $|Q| > n_k \geq  n(k,1)$, this implies that any point can be feasibly assigned to $I$. Let $p$ be the location of $s_k^{\ALG}$ during the phase. One possible way to serve all requests of this phase having $s_k$ at $p$ is to use $\cI$  (with possibly some initial cost of at most $W_{k-1}$ to bring the lighter servers in the right configuration). This gives that,
\begin{equation}
\label{eq:delta_heav_loc}
\Delta \SWF(p) \leq \cost(\cI) + W_{k-1} \leq w_k/3 + W_{k-1}.
\end{equation}
On the other hand, by Lemma~\ref{cl:phase_start_end}, during the phase
$\Delta\WF(p) \geq w_k/2 - 2W_{k-1}$, and by Lemma~\ref{lem:swf_wf_phase}, $ \Delta\SWF(p) \geq \Delta\WF(p) - W_{k-1} $. Together, this gives
$\Delta\SWF(p) \geq  w_k/2 - 3W_{k-1}$
which contradicts \eqref{eq:delta_heav_loc},  as $W_{k-1} \ll w_k/40k$.
\end{proof}

The next simple observation shows that if a point is not lucky during a phase, its work function value must be non-trivially high at the end of the phase.
\begin{observation}
\label{obs:non-lucky}
Consider a phase that starts at time $t$ and ends at $t'$. Let $p$ be a point which is not lucky during that phase. Then, $ \WF_{t'}(p) \geq M_t + w_k/ (5k \cdot n_k) $.
\end{observation}
\begin{proof}
By Lemma~\ref{lem:WFvsSWF} we have either $\WF_{t'}(p) \geq M_{t} + w_k$, in
which case the result is trivially true. 
Otherwise, we have that \[\WF_{t'}(p) \geq \WF_{t}(p) + \Delta_{t}^{t'}\SWF(p) - W_{k-1}.\] 
But as $p$ is not lucky, $\Delta \SWF(p) \geq w_k/(4kn_k)$, and as $W_{k-1} \leq w_k/(20 k \cdot n_k)$, together this gives  have that
$\WF_{t'}(p) \geq \WF_{t}(p) + w_k/ (5k \cdot n_k)$.
\end{proof}

\paragraph*{Proof of Lemma \ref{lem:main_gen}} We first give the upper bound on cost of $\ALG$ and then the lower bound on $\Delta M$.
\vspace{2mm}

{\em Upper Bound on cost of $\ALG$:} We denote by $\cost_i(\ALG)$ the cost of $\ALG$ during $i$th phase. Let $p_i$ be the location of $s_k^{\ALG}$, and $\Delta_i M$ the increase of $M$ during the $i$th phase. We will show that $\cost_i(\ALG) \leq 2 \cdot r_{k-1} \cdot w_k  + r_{k-1} \cdot \Delta_i M$. By summing over all $n_k+1$ phases, we get the desired upper bound.

 During the $i$th phase, $\ALG$ uses $k-1$ servers. Let $C_i^{k-1}$ denote the optimal cost to serve all requests of the $i$th phase starting at the same configuration as $\ALG$ and using only the $k-1$ lightest servers. By the inductive hypothesis of Theorem~\ref{thm:ub}, $\ALG$ using $k-1$ servers is strictly $r_{k-1}$-competitive, thus the cost incurred by $\ALG$ during the phase is at most $r_{k-1} \cdot C_i^{k-1} $. 

Now we want to upper bound $C_i^{k-1}$. By definition of static work function, there exists a schedule $S$ of cost $ \Delta\SWF(p_i)$ that serves all requests of the phase with $s_k$ fixed at $p_i$. Thus, a possible offline schedule for the phase starting at the same configuration as $\ALG$ and using only the $k-1$ lightest servers, is to move them at the beginning of the phase to the same locations as they are in $S$ (which costs at most $W_{k-1}$) and then simulate $S$ at cost $ \Delta\SWF(p_i)$. We get that $C_i^{k-1} \leq \Delta\SWF(p_i) + W_{k-1} $.

Moreover, $\ALG$ incurs an additional cost of $w_k$ for the move of server $s_k$ at the end of the phase. We get that 
\begin{equation}
\label{eq:alg_cost_phase}
\cost_i(\ALG) \leq w_k +  r_{k-1} \cdot(\Delta \SWF(p_i) + W_{k-1}).
\end{equation}
Combining this with Lemmas~\ref{cl:phase_start_end} and~\ref{lem:swf_wf_phase}, and using that $W_{k-1} \leq w_k /(20k \cdot n_k) $, we get 

\begin{align*}
\cost_i(\ALG) & \leq w_k +  r_{k-1} \cdot(\Delta \SWF(p_i) + W_{k-1}) \leq w_k +  r_{k-1} \cdot(\Delta \WF(p_i) + 2W_{k-1})  \\
 & \leq w_k +  r_{k-1} \cdot(\Delta_i M + w_k/2  + 4W_{k-1}) \leq w_k +  r_{k-1} \cdot(\Delta_i M + w_k/2  + (4/20) \cdot w_k) \\
 &  \leq 2 w_k \cdot r_{k-1} + r_{k-1} \cdot \Delta_i M. 
\end{align*}

{\em Lower bound on $\Delta M$:} Let $t_1$ and $t_2$ be the start and the end time of the $m = n_k+1$ phases. For the sake of contradiction, suppose that $M_t < M_{t_1} + w_k/(8 k \cdot n_k)$  for all $t \in [t_1,t_2]$.
By Lemma~\ref{lem:egalitarians_cost}, during first phase there are at most $n_k$ lucky points.
We claim that $s_k^{\ALG}$ must necessarily visit some lucky point in each subsequent phase. For $1 \leq i \leq m$, let $Q_i$ denote the set of points that have been lucky during all the phases $1,\dotsc,i$.
Let $t$ denote the starting time of $i$th phase and $p$ the location of $s_k^{\ALG}$ during this phase, for any $i \geq 2$. By Lemma~\ref{cl:phase_start_end}, we have that
\begin{align*}
\WF_t(p) < M_t + W_{k-1}   \leq M_{t_1} + w_k / (5 k \cdot n_k).
\end{align*}
By Observation \ref{obs:non-lucky}, this condition can only be satisfied by
points $p \in Q_{i-1}$ and hence we get that $p$ was lucky in all previous
phases. Now, by Lemma~\ref{cl:phase_start_end}, $i$th phase $\WF(p)$
rises by at least $w_k/2 - 2W_{k-1}$ during the $i$th phase,
and hence $p$ is not lucky. Therefore, $p \notin Q_i$ and $p \in Q_{i-1}$
and $|Q_{i}| \leq |Q_{i-1}| - 1$.
Since $|Q_1| \leq n_k = m-1$, we get that $Q_m=\emptyset$, which gives the
desired contradiction.
\qed

\subsection{Upper bound for $d$ different weights}
\label{sec:ubd}
For the case of $d$ weight classes
we prove a more refined upper bound.
The general approach is quite similar as before.
However, the proof of the variant of Lemma~\ref{lem:egalitarians_cost}
for this case is more subtle as the number of ``lucky"
locations for the heaviest servers can be infinite. However, we handle this
situation by maintaining posets of lucky tuples. We show that it suffices for
$\ALG$ to traverse all the minimal elements of this poset, and we use
Dichotomy theorem~\ref{thm:dichotomy-d} to bound the number of these minimal
elements.

\paragraph*{Definitions and Notation.}
First, we need to generalize a few definitions which were used until now.
Let $w_1 < \dotsb < w_d$ be the weights of the servers, where
$k_i$ is the number of servers of weight $w_i$, for $i=1, \dotsc, d$.
Henceforth, we assume that the values of $k_1, \dotsc, k_d$ are fixed,
as many constants and functions in this section will implicitly depend on
them. For example, $r_{d-1}$ denotes
the competitive ratio of $\ALG$ with servers of $d-1$ different weights, and
it depends on $k_1, \dotsc, k_{d-1}$.

We denote $W_i = \sum_{j=1}^i w_j k_j$,
and we assume $W_{d-1} \leq w_d/(20kn_d)^{k_d}$, where $n_d$ equals to the value
of $n(d,k_d)$ from Dichotomy theorem~\ref{thm:dichotomy-d}.
This assumption can not affect the competitive ratio by
more than a factor $(20kn_d)^{dk_d}$, what is smaller than our targeted ratio.
We also assume that the universe of pages $U$ contains at least $k$ pages that
are never requested. This assumption is only for the purpose of the
analysis and can be easily satisfied by adding artificial pages to $U$, without
affecting the problem instance.

A configuration of servers is a function
$C\colon \{1, \dotsc, d\} \to 2^U$, such that
$|C(i)| = k_i$ for each $i$.
Servers with the same weight are not distinguishable and we manipulate them in
groups. Let $K_i$ denote the set of servers of weight $w_i$.
For a $k_d$-tuple $A_d$,
we define the minimum work function value over all configurations having the
servers of $K_d$ at $A_d$, i.e.~$\WF_t(A_d) = \min\{\WF_t(C) \;|\, C(d) = A_d\}$.
Similarly, we define $\SWF_t(A_d)$ the {\em static work function} at time $t$
as the optimal cost of serving the requests $\sigma_1,\dotsc, \sigma_t$ while
keeping the servers of $K_d$ fixed at $A_d$.
When calculating the value of the work function and the static work function,
we use the distance $k_i\,w_i$ whenever the optimal solution moved at least one server from
$K_i$.
This work function still estimates the offline
optimum with a factor $k$, and is easier to work with.

As in previous subsection, we use $\Delta_{t_1}^{t_2} X$ to denote the change in
quantity $X$ during time interval $[t_1,t_2]$.  We also use the function
$n(d,t)$ from Theorem~\ref{thm:dichotomy-d}.
Observe that $n(d,t) \leq n_d$ for all $1\leq t \leq k_d$.

\paragraph*{Algorithm.}
We prove the bound for $\WFA_{0.5}$ with slightly
deformed distances between the configurations. More precisely,
we define $d(A,B) = \sum_{i=1}^d k_iw_i \mathbf{1}_{(A(i) \neq B(i))}$,
and denote $\ALG$ the $\WFA_{0.5}$ with this distance function.
In particular, $\ALG$ chooses new configuration for its heaviest
servers without distinguishing between those which differ only in a single
position and those which are completely different.
We call a \textit{phase} the maximal time interval when $\ALG$
does not move \textit{any} server from $K_d$.

\paragraph*{Basic Properties of WFA.}
Here are a few simple properties whose proofs are not very interesting
and are contained in Appendix~\ref{app_sec:ub}.

\begin{lemma}\label{lem:d_phase_start_end}
Consider a phase that starts at time $t_1$ and finishes at time $t_2$. Let $C_d$
be the $k_d$-tuple where the algorithm has its heaviest servers $K_d$ during the
phase. Then,
\vspace{-1ex}
\begin{enumerate}[(i)]
\item $M_{t_1}\leq \WF_{t_1}(C_d) \leq M_{t_1} + W_{d-1} $, and
\item $k_d\,w_d/2 - 2W_{d-1} \leq \Delta\WF(C_d)
	\leq \Delta M + k_d\,w_d/2 + 2W_{d-1}$.
\end{enumerate}
\end{lemma}

\begin{lemma}
\label{lem:swf_wf_dlev}
For a phase where $\ALG$ has its servers from $K_d$ at a $k_d$-tuple $C_d$,
we have
\[ \Delta\WF(C_d) - W_{d-1} \leq \Delta\SWF(C_d) \leq \Delta\WF(C_d) + W_{d-1}.
\]
\end{lemma}

\begin{lemma}\label{lem:WFvsSWFd}
Let $M_t$ be the minimum value of work function at time $t$. For $t'>t$ and any
$k_d$-tuple $C_d$, we have the following:
\[ \WF_{t'}(C_d) \geq \min\{\WF_t(C_d) + \Delta_{t}^{t'}\SWF(C_d)
	- W_{d-1}, M_t + w_d\}.
\]
\end{lemma}

\paragraph*{Main Lemma.}
The following lemma already implies a competitive ratio of order
$n_d^{O(k_d)}\,r_{d-1}$.

\begin{lemma}\label{lem:main-d}
Let us consider a group of $a_d = (k_d^3\,n_d)^{k_d}$ consecutive phases.
We have, $\Delta M \geq w_d/(10kn_d)^{k_d}$ and
$\cost(\ALG) \leq 2a_d\,r_{d-1} k_d\,w_d + r_{d-1} \Delta M$,
where $r_{d-1}$ is the strict competitive ratio of $\ALG$ with servers
$K_1, \dotsc, K_{d-1}$.
\end{lemma}

The bound for $\cost(\ALG)$ is easy and can be shown using a combination of the
basic properties mentioned above.
Therefore, most of this section focuses on lower bounding $\Delta M$.

Let $t_1$ and $t_2$ denote the beginning and the end of this group of phases.
At each time $t$, we maintain a structure containing
all configurations $C_d$ for the servers in $K_d$ such that
$\WF_t(C_d)$ could still be below $M_{t_1} + w_d/(10kn_d)^{k_d}$.
We call this structure a {\em poset of lucky tuples} and it is defined below.
Then, we show that this poset gets smaller with each phase until it becomes
empty before time $t_2$.

\paragraph{Poset of lucky tuples.}
Let us first consider a single phase. We call a $k_d$-tuple $C_d$ {\em lucky},
if we have $\Delta\SWF(C_d) < w_d/(4kn_d)^{k_d}$ during this phase.
A tuples $T$ of size $t<k_d$ is called lucky,
if $\Delta\SWF(C_d) <w_d/(4kn_d)^t$ for
each $k_d$-tuple $C_d$ containing $T$.
Let $\cQ_i$ be the set of tuples which were lucky during phase $i$.
We denote $(\cL_i, \subseteq) = \bigcup_{T\in \cQ_i} \cl(T)$ and we call it
the poset of lucky tuples during the phase $i$. Here, the closure $\cl(T)$ is a
set of all tuples of size at most $k_d$ which contain $T$ as a subset.
The following lemma bounds the number of its minimal
elements and uses Dichotomy theorem~\ref{thm:dichotomy-d}.

\begin{lemma}\label{lem:dipaying}
Let us consider a poset $\cL$ of tuples which are lucky during one phase,
and let $\cE_t$ the set of its minimal elements of size $t$.
Then we have $|\cE_t| \leq n(d,t)$.
\end{lemma}

The following observation show that if a $k_d$-tuple was unlucky during at
least one of the phases, its work function value must already be
above the desired threshold.

\begin{observation}\label{obs:unluckyd}
Let us consider a phase between times $t$ and $t'$.
If a $k_d$-tuple $C_d$ was not lucky during this phase,
we have $\WF_{t'}(C_d) \geq M_{t} + w_d/(5kn_d)^{k_d}$.
\end{observation}
\begin{proof}
By Lemma~\ref{lem:WFvsSWFd}, we have either $\WF_{t'}(C_d) \geq M_{t} + w_d$,
in which case the result trivially holds, or
\[ \WF_{t'}(C_d) \geq \WF_{t}(C_d) + \Delta_{t}^{t'}\SWF(C_d).\]
But then  we have that
$\WF_{t'}(C_d) \geq \WF_{t} + w_d/(4kn_d)^{k_d} -W_{d-1}$,
as $C_d$ was unlucky, and this is at least $w_d/(5kn_d)^{k_d}$ as
$W_{d-1} \leq w_d/(20kn_d)^{k_d}$.
\end{proof}

Therefore, we keep track of the tuples which were lucky in all the phases.
We denote $\cG_m = \bigcap_{i=1}^m \cL_i$ the poset of tuples which were lucky
in each phase $1, \dotsc, m$.
Note that we can write $\cG_m = \bigcup_{T\in\cE} \cl(T)$, where $\cE$ is the
set of the minimal elements of $\cG_m$.
If, in phase $m+1$, we get $\cL_{m+1}$ which does not contain some
$\cl(T)\subseteq \cG_m$, then $\cl(T)$ might break into closures of some
supersets of $T$. This is a favourable situation for us because it makes
$\cG_{m+1}$ smaller than $\cG_m$.
The following lemma claims that $\cl(T)$ cannot break into too many
pieces.

\begin{lemma}\label{lem:Td_decompose}
Let $T$ of size $t$ be a fixed minimal tuple in $\cG_{m}$.
If $\cl(T) \nsubseteq \cL_{m+1}$, then
$\cl(T) \cap \cL_m$ contains no tuple of size $t$
and, for $i=1, \dotsc, k_d-t$, it contains at most
$k_d\,n_d$ tuples of size $t+i$.
\end{lemma}
\begin{proof}
Let $T'$ be some inclusion-wise minimal tuple from $\cL_m$.
It is easy to see that $\cl(T)\cap\cl(T') = \cl(T\cup T')$,
and $T\cup T'$ is the new (potentially) minimal element.
Denoting $\cE$ the set of minimal elements in $\cL_m$, we have
$\cl(T) \cap \cL_m = \bigcup_{T'\in \cE} \cl(T)\cap \cl(T')$.
Therefore, $\cl(T) \cap \cL_m$ contains
at most one minimal element per one minimal tuple from $\cL_m$.

Let us now consider the resulting $T\cup T'$ according to its size.
The size of $T\cup T'$ can be $t+i$ if the size of $T'$ is at least $i$ and at
most $t+i$. Therefore, by Lemma~\ref{lem:dipaying},
we have at most $\sum_{j=i}^{t+i} n(d,j) \leq k_d\,n_d$ minimal elements of size
$t+i$.
\end{proof}

\paragraph{Proof of the main lemma.}
First, let us bound the cost of the algorithm.
During phase $i$ when its heaviest servers reside in $C_d^i$, it
incurs cost $\cost_i(\ALG) \leq k_d\,w_d+ r_{d-1}(\Delta\SWF(C_d^i)+W_{d-1})$.
The first $k_d\,w_d$ is the cost for the single move of servers in $K_d$ at the
end of the phase, and we claim that the second term is due to the movement of
the servers $K_1, \dotsc, K_{d-1}$.

To show this, we use the assumption that $\ALG$ is strictly
$r_{d-1}$-competitive when using servers $K_1, \dotsc, K_{d-1}$.
Let us denote $C^i_1, \dotsc, C^i_{d-1}$ their configuration at the beginning of
the phase. The servers from $K_1, \dotsc, K_{d-1}$
have to serve the request sequence
$\bar{\sigma}^i$, consisting of all requests issued during the phase which do
not belong to $C_d^i$, starting at configuration $C^i_1, \dotsc, C^i_{d-1}$.
We claim that there is such offline solution with cost
$\Delta\SWF(C_d^i)+W_{d-1}$: the solution certifying the value of $\SWF(C_d^i)$
has to serve the whole $\bar{\sigma}^i$ using only $K_1, \dotsc, K_{d-1}$,
although it might start in a different initial position, and therefore we need
additional cost $W_{d-1}$.

Therefore, the cost incurred by $\ALG$ during the phase $i$ is at most
$k_d\,w_d+ r_{d-1}(\Delta\SWF(C_d^i)+W_{d-1})$.
Combining lemmas \ref{lem:d_phase_start_end} and \ref{lem:swf_wf_dlev}, we get
$\Delta\SWF(C_d^i) \leq \Delta_i M + k_d\,w_d/2 + 3W_{d-1}$, and summing this
up over all phases, we get
\begin{align*}
\cost(\ALG)
	&\leq a_d\,k_d\,w_d +  r_{d-1}
		(\Delta M + a_d\cdot k_d\,w_d/2 + a_d\cdot 4W_{d-1})\\
	&\leq a_d\,k_d\,w_d + r_{d-1}\,a_d\,k_d\,w_d/2
		+ r_{d-1}\,a_d\,4W_{d-1} + r_{d-1}\Delta M
	\leq 2a_d\,r_{d-1}\,k_d\,w_d + r_{d-1}\Delta M,
\end{align*}
since $4W_{d-1} < w_d/2$.

Now we bound $\Delta M$.
Clearly, if $M_t \geq M_{t_1} + w_d/(10kn_d)^{k_d}$ for some $t\in [t_1, t_2]$,
we are done.
Otherwise, we claim that the posets $\cG_i$ become smaller with each phase and
become empty before the last phase ends. We define a potential which captures
their size:
\[
\Phi(i) = \sum_{j=1}^{k_d} (2k_d)^{k_d-j} \cdot (k_d\,n_d)^{k_d-j}
	\cdot L_j(i),
\]
where $L_j(i)$ is the number of minimal $j$-tuples in $\cG_i$.

Let $t$ and $t'$ denote the beginning and the end of the $i$th phase,
and $A_d$ be the configuration of $K_d$ during this phase.
By Lemma~\ref{lem:d_phase_start_end},
we have $\WF_t(A_d) \leq M_t+W_{d-1} < M_{t_1} + w_d/(10kn_d)^{k_d} +W_{d-1}$
and $\Delta_t^{t'}\WF(A_d) \geq k_d\,w_d/2 - 2W_{d-1}$.
By Observation~\ref{obs:unluckyd}, this implies that $A_d$ belongs to
$\cG_{i-1}$ and does not belong to $\cG_i$.
Therefore, at least one $\cl(T) \subseteq \cG_{i-1}$
(the one containing $A_d$) must have broken during phase $i$.

Each $\cl(T)$ that breaks into smaller pieces causes a change of the potential,
which we can bound using Lemma~\ref{lem:dipaying}. We have
\[ \Delta \Phi \leq
	- (2k_d)^{k_d-|T|} \, (k_d\, n_d)^{k_d-|T|}
	+ (2k_d)^{k_d-(|T|+1)} \,
		(k_d\,n_d)^{k_d-(|T|+1)} \cdot k_d\cdot k_d\,n_d.
\]
The last term can be bounded by
$k_d (2k_d)^{k_d-(|T|+1)} (k_d\, n_d)^{k_d-|T|}$, what is strictly smaller than
$(2k_d)^{k_d-|T|} \cdot (k_d\, n_d)^{k_d-|T|}$.
So, we have $\Delta\Phi \leq -1$, since the value of $\Phi(i)$ is always
integral.

The value of $\Phi$ after the first phase is
$\Phi(1) \leq k_d \cdot \big((2k_d)^{k_d} (k_d\,n_d)^{k_d-1} \cdot n_d\big)
< a_d$,
by Lemma~\ref{lem:dipaying}, and $\Phi(i)$ becomes zero as soon as $\cG_i$ is
empty. Therefore, no page can be lucky during the entire group of $a_d$ phases.
\qed

\paragraph{Proof of Lemma~\ref{lem:dipaying}.}
We proceed by contradiction.
If the lemma is not true for some $t$,
then there exists a set of $t$-tuples $\cQ_t\subseteq \cE_t$
of size $n(d,t)+1$.
For each $T\in \cQ_t$, we consider a service pattern $\cI^T$ which is
chosen as follows.
For a $k_d$-tuple $A_T$ containing $T$ and $k_d-t$ points which were not
requested during the phase, we have $\Delta\SWF(A_d) < w_d/(4kn_d)^t$.
Therefore there is a service pattern $\cI^T$ of cost smaller than
$w_d/(4kn_d)^t$ such that $T$ is a feasible label for its top-level interval.

We consider a common refinement $\cI$ of all service patterns $\cI^T$, for $T \in Q_t$.
Its cost is less than $k\sum_{T\in\cQ_t} \cost(\cI^T)$, and each
$T\in \cQ_t$ is a feasible label for its single top-level interval $I$.
Common refinement $\cI$ has more than $n(d,t)$ minimal feasible $t$-tuples,
so by Theorem~\ref{thm:dichotomy-d}, $Q_1=U$.
This implies that the configuration $A_d$ of the heaviest servers of $\ALG$
during this phase
is also feasible label for $I$, and therefore
\[
\Delta\SWF(A_d) \leq \cost(\cI^T) + W_{d-1}
	< k(n_d+1) w_d/(4kn_d)^t + W_{d-1}
	\leq \frac14 (1+1/n_d) \cdot \frac{w_d}{(4kn_d)^{t-1}} + W_{d-1}.
\]
This is smaller than $w_d/(4kn_d)^{t-1}$, because $W_{d-1}$ is less than
$w_d/(20kn_d)^{k_d}$.
However, lemmas \ref{lem:d_phase_start_end} and \ref{lem:swf_wf_dlev}
imply that
$\Delta\SWF(A_d) \geq w_d/2  - W_{d-1}$, what gives a contradiction.
\qed

\paragraph*{Proof of Theorem~\ref{thm:ubd}}
We prove the theorem by induction on $d$. For $d=1$ we have the classical paging
problem and it is known that $\ALG$ is
$O(k_1)$-competitive, see \cite{ST85}.

{\em Competitive ratio.}
Since $\cost(\OPT) = M_T/k$, it is enough to compute the ratio between
$\cost(\ALG)$ and $\Delta M$ during one group of $a_d$ phases, getting $1/k$
fraction of the ratio.
The case where the last group contains less than $a_d$ phases can be handled similarly as in proof of Theorem~\ref{thm:ub}.
By the main lemma~\ref{lem:main-d}, we get the following recurrence.
\begin{equation}\label{eq:R_d-rec}
\frac1k r_d \leq \frac{2a_d\,r_{d-1}\,k_d\,w_d}{w_d/(10kn_d)^{k_d}}
	+ \frac{r_{d-1} \Delta M}{\Delta M}
	\leq a_d^3 r_{d-1}.
\end{equation}

{\em Strict competitive ratio.}
It is enough to the same case analysis as in the proof of Theorem~\ref{thm:ub}.
Applying the corresponding variants of the technical lemmas
(\ref{lem:d_phase_start_end}, \ref{lem:swf_wf_dlev}),
it can be shown that in all of those cases, the ratio between the cost of the
algorithm and the cost of the adversary is much smaller than $a_d^3\,r_{d-1}$.

{\em Calculating the recurrence.}
Let us assume that $r_{d-1} \leq 2^{12dk^3 \prod_{j=1}^{d-1}(k_j+1)}$,
and recall that $a_d = (k_d^3\,n_d)^{k_d}$, where
$n_d = n(d,k_d)$ where
$\log n(d,k_d) \leq 4dk^2 k_d \prod_{j=1}^{d-1}(k_j+1)
	\leq 4dk^3\prod_{j=1}^{d-1}(k_j+1)$.
Therefore, taking the logarithm of \eqref{eq:R_d-rec}, we get
\[ \log r_d \leq \log k + 9k_d \log k_d
	+ 3k_d\cdot 4dk^3 \prod_{j=1}^{d-1}(k_j+1)
	+ 12dk^3 \prod_{j=1}^{d-1}(k_j+1).
\]
The last two terms are quite similar, and we can bound them by
$(k_d+1)\cdot 12dk^3 \prod_{j=1}^{d-1}(k_j+1)$.
Moreover, the first two terms are
smaller than $12k^3$. Therefore we get the final bound
\[\pushQED{\qed}
r_k \leq 2^{(k_d+1)\cdot 12(d+1)k^3 \prod_{j=1}^{d-1}(k_j+1)}
	\leq 2^{12(d+1)k^3 \prod_{j=1}^d(k_j+1)}.\qedhere
\popQED
\]

\section{Concluding Remarks}
\label{sec:conc}

There are several immediate and longer-term research directions.
First, it seems plausible that using randomization a singly exponential (i.e.~logarithmic in the deterministic bound) competitive ratio against oblivious adversaries can be achieved. We are unable to show this, since our loss factor from Lemma~\ref{lem:egalitarians_cost} is much higher due to the refinement technique.

Another natural question is to consider weighted $k$-server for more general metrics. As discussed in Section~\ref{sec:rel_work}, nothing is known even for the line beyond $k=2$. Obtaining any upper bound that is only a function of $k$ would be very interesting, as it should lead to interesting new insights on the generalized work-function algorithm (which seems to be the only currently known candidate algorithm for this problem).

Finally, the generalized k-server problem, described in Section~\ref{sec:rel_work}, is a far reaching generalization of the weighted $k$-server problem for which no upper bound is known beyond $k=2$, even for very special and seemingly easy cases. For example, when all metrics are uniform, Koutsoupias and Taylor \cite{KT04} showed a lower bound of $2^k-1$, but no upper bounds are known. 
We feel that exploring this family of problems should lead to very interesting techniques for online algorithms.

\section*{Acknowledgments}
We are grateful to Ren\'e Sitters for first bringing the problem to our attention.
We would like to thank Niv Buchbinder, Ashish Chiplunkar and Janardhan Kulkarni for several useful discussions during the initial phases of this project.
Part of the work was done when NB and ME were visiting the Simons Institute at Berkeley and we thank them for their hospitality.

\newpage

\appendix

\section{Lower Bound for general metric spaces}
\label{sec:general_lb}
We now show that our lower bound from Theorem~\ref{thm:lb} naturally extends to
any metric space. We use the same notation for constants and strategies as in
Section~\ref{sec:lb}.

\paragraph{High-level idea.}
Our strategy consists of an arbitrary number of executions of the strategy
$S_{k-1}$.
We define $n_k$ adversaries, each having $s_k$ at a different page, and
we compare $\ALG$ to their average cost.
Recall that $n_k \geq 2^{2^{k-4}}$.

\begin{theorem}
Let $(U,d)$ be an arbitrary metric space with at least $n_k+1$ points.
No deterministic algorithm for the weighted $k$-server problem can be
better than $\Omega(2^{2^{k-4}})$-competitive on $U$.
\end{theorem}
\begin{proof}
Let $\ALG$ be a fixed algorithm. We choose a set $P \subseteq U$ of $n_k+1$
points. Without loss of generality, the minimum distance between two points of
$P$ is 1, and we denote $D$ the maximum distance.

In the constructed instance, the weights of the servers are chosen as follows:
$w_1 = 1$, and $w_i = n_k \cdot D \cdot \sum_{j=1}^{i-1} w_j$,
for $2 \leq i \leq k$.
Let $\cost(s_j^{\ADV_i})$ denote the cost due to moves of server $s_j$ of
adversary $\ADV_i$. Similarly, $\cost(s_j^{\ALG})$ denotes the cost of server
$s_j^{\ALG}$. Let $A_{k-1} = \sum_{i=1}^{n_k} \sum_{j=1}^{k-1}
\cost(s_j^{\ADV_i})$ denote the total cost incurred by the $k-1$ lighter servers
of the adversaries. This way, we have
\begin{equation}
\label{eq:sum_adv}
\sum_{i=1}^{n_k} \cost(\ADV_i) = A_{k-1} + \sum_{i=1}^{n_k} \cost(s_k^{\ADV_i}).
\end{equation} 

We maintain the following invariant: at any given time, $\ALG$ and each of the
adversaries $\ADV_1,\dotsc,\ADV_{n_k}$ have their heaviest server $s_k$ at a
different point of $P$. This way, for each point $p \in P$ either $\ALG$ or some
adversary has its heaviest server at $p$. To achieve this, initially
all adversaries move $s_k$ to a different point. The cost of those moves
is fixed and does not affect the competitive ratio, so we can ignore it. Then,
whenever $\ALG$ moves $s_k^{\ALG}$ from point $p$ to $q$, the adversary which
has its heavy server at $q$ moves it to $p$. The adversaries do not move their
heaviest server $s_k$ at any other time. This way we ensure that
\begin{equation}
\label{eq:advs_sk}
\sum_{i=1}^{n_k} \cost(s_k^{\ADV_i}) = \cost(s_k^{\ALG}).
\end{equation}

It remains to show that we can create a request sequence, such that the cost of
moves of $k-1$ lighter servers of all adversaries is at most cost of $\ALG$,
i.e. $A_{k-1} \leq \cost(\ALG)$.

We create the request sequence using the adaptive strategies defined in
Section~\ref{sec:lb}. The whole sequence consits of arbitrary number of
executions of $S_{k-1}$: if $s_k^{\ALG}$ is located at $p\in P$,
we run the strategy $S_{k-1}(P\setminus \{p\})$. 
Whenever $s_k^{\ALG}$ moves from $p$ to $q$, we terminate the current 
execution of $S_{k-1}(P\setminus \{p\})$ and start
$S_{k-1}(P\setminus \{q\})$.

Each execution of $S_{k-1}(P\setminus\{p\})$ we call a phase. For each phase, we
create a service pattern $\cI$ as in Section~\ref{sec:lb}.
Clearly, $\cI$ has only one $k$th level interval denoted by $I$.
Due to Lemma~\ref{lem:pages_useful}, any point $q \in P\setminus\{p\}$
is a feasible label for $I$. Since each $s_k^{\ADV_i}$ is located at some point
from $P\setminus \{p\}$, all
adversaries can serve the requests of the phase using the service pattern $\cI$.

By Lemma~\ref{lem:prta}, we know that each adversary does not move a server
$s_j$, unless $\ALG$ moves a server $s_i^{\ALG}$, for $i>j$.
Therefore, whenever $\ALG$ moves server $s_i$, the total cost incured by
all adversaries is at most $n_k \cdot D\sum_{j=1}^{i-1} w_j$, which is at most $w_i$ thanks to the weight separation.
Therefore, we get
\begin{align}\label{eq_advs_k-1}
A_{k-1} \leq \cost(\ALG).
\end{align}

Now, combining \eqref{eq:sum_adv},\eqref{eq:advs_sk} and\eqref{eq_advs_k-1},
we get
\begin{align*}
\sum_{i=1}^{n_k} \cost(\ADV_i) &= A_{k-1} + \sum_{i=1}^{n_k} \cost(s_k^{\ADV_i}) \leq \cost(\ALG) + \cost(s_k^{\ALG}) \leq 2 \cdot \cost(\ALG),
\end{align*} 
which implies the lower bound of $n_k/2$.
\end{proof}

\section{Omitted proofs from Section~\ref{sec:ubs}\label{app_sec:ub}}
Here we prove the basic properties of $\WF$ and $\SWF$.
We prove them explicitly only for the general case of arbitrary weights,
however, the proofs can be adapted easily to the case of $d$ different weights
by replacing $w_k$ by $k_d\,w_d$ and $W_{k-1}$ by $W_{d-1}$.

{\bf Lemma~\ref{cl:phase_start_end}. } {\em
Consider a phase that starts at time $t_1$ and finishes at time $t_2$. Let $p$
be the point where the algorithm has its heaviest server $s_k^{\ALG}$ during the
phase. Then,
\begin{enumerate}[(i)]
\item $M_{t_1}\leq \WF_{t_1}(p) \leq M_{t_1} + W_{k-1} $, and
\item $ w_k/2 - 2W_{k-1} \leq \Delta \WF(p) \leq \Delta M + w_k/2 + 2 W_{k-1}.$ 
\end{enumerate}
}

\begin{proof}
{\em (i)}
The fact that $M_{t_1}\leq \WF_{t_1}(p)$ is obvious, as $M_{t_1}$ is the minimum
work function value at time $t_1$. It remains to show that
$\WF_{t_1}(p) \leq M_{t_1} + W_{k-1}$.

Let us suppose that $\ALG$ moved from configuration $A$ to $B$.
Since it is the beginning of this phase, we have $A(k)\neq B(k) = p$
and therefore $d(A,B) \geq w_k$. Let $C$ be a configuration such that
$\WF_{t_1}(C) = M_{t_1}$. Surely, $d(A,C) \leq w_k + W_{k-1}$ and, since $\ALG$
prefered to move to $B$ instead of $C$, we get
\[ \WF_{t_1}(B) + w_k/2 \leq \WF_{t_1}(C) + (w_k+W_{k-1})/2, \]
and therefore $\WF_{t_1}(p) \leq \WF_{t_1}(B) \leq M_{t_1} + W_{k-1}$.

{\em (ii)}
First, we show that $\WF_{t_2}(p) \leq M_{t_2} + w_k/2 + 2W_{k-1}$. Together
with {\em (i)}, this implies $\Delta\WF(p) \leq M_{t_2}-M_{t_1} + w_k/2 +
2W_{k-1}$. For any time $t \in (t_1,t_2)$, if we have $\WF_t(p) > M_t + w_k/2 + W_{k-1}/2$, then $\ALG$ would prefer to
move to some configuration $C$ with $\WF_{t}(C) = M_t$.
Therefore, at time $t'= t_2-1$ we have 
\[ \WF_{t'}(p) \leq M_{t'} + w_k/2 + W_{k-1}/2. \]
Given a request, the value of $\WF(p)$ can increase by at most $w_1$ (since one possible way to serve the request is by using the lightest server), therefore  $\WF_{t_2}(p) \leq M_{t'} + w_k/2 + W_{k-1}/2 + w_1 \leq M_{t_2} + w_k/2 + 2W_{k-1}$.

To get the lower bound for $\Delta\WF(p)$, we claim that
$\WF_{t_2}(p) \geq M_{t_2} + w_k - W_{k-1}$. Together with {\em (i)}, this
already implies $\Delta\WF(p) \geq \Delta M + w_k/2 - 2W_{k-1}$.
Suppose we had $\WF_{t_2}(p) < M_{t_2} + w_k/2 - W_{k-1}$.
Let $A$ be the configuration of $\ALG$ and let $B$ be a configuration such that
$B(k)=p$ and
$\WF_{t_2}(B) < M_{t_2} + w_k/2 - W_{k-1}$. Since $d(A,B) \leq W_{k-1}$, we have
\[ \WF_{t_2}(B) + d(A,B)/2 < M_{t_2} + w_k/2. \]
On the other hand, for any configuration $C$ such that $C(k)\neq p$,
we have $\WF_{t_2}(C)+d(A,C) \geq M_{t_2} + w_k/2$, what means that $\ALG$ would
not prefer to move $s_k$ at time $t_2$, a contradiction.
\end{proof}

\medskip\noindent
{\bf Lemma~\ref{lem:swf_wf_phase}. } {\em
For a phase of $\ALG$, where $s_k^{\ALG}$ is at point $p$, we have that
\begin{enumerate}[(i)]
\item $\Delta\SWF(p) \geq \Delta \WF(p) - W_{k-1}$, and
\item $\Delta\SWF(p) \leq \Delta\WF(p) +  W_{k-1}. $ 
\end{enumerate}
} 
\begin{proof}
{\em (i)} Let $\cI$ be the optimal feasible service pattern to serve all requests of the phase, with a single $k$th level interval
assigned to $p$. Clearly, $\Delta \SWF(p) \geq \cost(\cI)$. Also, $ \Delta
\WF(p) \leq \cost(\cI) + W_{k-1} $, since one possible way to serve requests of
the phase is to use some feasible labeling $\alpha $ of $\cI$: it costs at most
$W_{k-1}$ to move to the initial configuration of $\alpha$, and then serve
requests according to $\alpha$, paying $\cost(\cI)$. Combining this,
we get $\Delta \WF(p) - \Delta \SWF(p) \leq W_{k-1}$.

{\em (ii)} Let $t_1$ and $t_2$ be start and end time of this phase, and $M_t$
the minimum work function value at time $t$. Consider the solution $P$ that
determines the value of $\WF_{t_2}(p)$. We claim that, according to $P$, $s_k$
stays at $p$ during the entire interval $[t_1, t_2]$.
This implies
\begin{align*}
\WF_{t_2}(p) & \geq \WF_{t_1}(p)  +\Delta\SWF(p) - W_{k-1},
\end{align*} 
and thus $ \Delta\SWF(p) \leq \Delta\WF(p) + W_{k-1}$. 

Thus it remains to show that, according to $P$, $s_k$ stays at $p$ during
$[t_1, t_2]$. For contradiction, let $t\in [t_1, t_2]$ be the
last time when $P$ moved $s_k$ to $p$. Then,
\begin{equation}
\label{eq:wf_phase}
\WF_{t_2}(p) \geq M_t + w_k + \Delta_{t}^{t_2} \SWF(p) - W_{k-1}.
\end{equation}
The term $W_{k-1}$ is because the $k-1$ lighter servers of the state defining $M_t$ at time $t$ could be at different locations than in $P$.
Moreover, by the definition of work-function,
\begin{equation}
\label{eq:dwf_phase}
 \WF_{t_2}(p) \leq \WF_t(p) + \Delta_{t}^{t_2}\SWF(p) + W_{k-1}.
\end{equation}
Combining \eqref{eq:wf_phase} and \eqref{eq:dwf_phase},
we get $\WF_t(p) \geq M_t + w_k - 2 W_{k-1}$.
However, by construction of $\ALG$, we have
$\WF_t(p) < M_t + w_k/2 + W_{k-1}$, which is a contradiction, since
$W_{k-1} \leq w_k/(20kn_k)$.
\end{proof}

\medskip\noindent
{\bf Lemma~\ref{lem:WFvsSWF}. } {\em  Let $M_t$ be the minimum value of work function at time $t$. For $t'>t$ and any
$p \in U$ we have the following:
\[ \WF_{t'}(p) \geq \min\{\WF_t(p) + \Delta_{t}^{t'}\SWF(p)
	- W_{k-1}, M_t + w_k\}.
\]
}
\vspace{-6ex}
\begin{proof}
Let us consider the offline optimal schedule serving the requests $\sigma_1,\dotsc, \sigma_{t'}$ and ending at a configuration $C'$ such that $C'(k) = p$. Let $C_t$ denote the configuration of the servers at time $t$ according to that schedule. Since $M_t$ was the
minimum work function value at time $t$ over all possible states, we have
$\WF_t(C_t) \geq M_t$.
There are two cases to consider:
\begin{itemize}
\item $C_t(k) \neq p$: Then $\WF_{t'}(p) \geq M_t + w_k$, because $s_k$ has to move
	to $p$ until time $t'$.
\item $C_t(k) = p$: If $s_k$ moved during this time, then
	$\WF_{t'}(p) \geq M_t + w_k$, otherwise
	$\WF_{t'}(p) \geq \WF_t(p) + \Delta\SWF(p) - W_{k-1}$.
\end{itemize}
\end{proof}

\end{document}